\documentclass[reprint,aps,pra,dcolumn,superscriptaddress,notitlepage,nofootinbib]{revtex4-2}
\usepackage{graphicx}
\usepackage{indentfirst}
\usepackage[top=2.5cm,left=2.0cm,right=2.0cm,bottom=2.5cm]{geometry}

\usepackage{amssymb}
\usepackage{array}
\usepackage{amsmath}
\usepackage{dsfont}
\usepackage{physics}
\usepackage{amsfonts}
\usepackage{cancel}
\usepackage{comment}
\usepackage{bm}
\usepackage{bbm}
\usepackage{url}
\usepackage{hyperref}
\usepackage{multirow}
\usepackage{subfig}
\usepackage{xcolor}
\usepackage[many]{tcolorbox}
\usepackage[english]{babel}
\usepackage{algorithm}
\usepackage[noend]{algpseudocode}
\usepackage{amsthm}
\usepackage{thmtools}
\usepackage{thm-restate}

\makeatletter
\newsavebox{\@brx}
\newcommand{\llangle}[1][]{\savebox{\@brx}{\(\m@th{#1\langle}\)}%
  \mathopen{\copy\@brx\mkern2mu\kern-0.9\wd\@brx\usebox{\@brx}}}
\newcommand{\rrangle}[1][]{\savebox{\@brx}{\(\m@th{#1\rangle}\)}%
  \mathclose{\copy\@brx\mkern2mu\kern-0.9\wd\@brx\usebox{\@brx}}}
\makeatother

\newcommand{\tikzxmark}{%
\tikz[scale=0.23] {
    \draw[line width=0.7,line cap=round] (0,0) to [bend left=6] (1,1);
    \draw[line width=0.7,line cap=round] (0.2,0.95) to [bend right=3] (0.8,0.05);
}}

\newcommand{\tikzcmark}{%
\tikz[scale=0.23] {
    \draw[line width=0.7,line cap=round] (0.25,0) to [bend left=10] (1,1);
    \draw[line width=0.8,line cap=round] (0,0.35) to [bend right=1] (0.23,0);
}}

\newtheorem{lemma}{Lemma}

\begin{document}

\title{Evaluating a quantum-classical quantum Monte Carlo algorithm with Matchgate shadows}

\author{Benchen Huang}
\affiliation{AWS Worldwide Specialist Organization, Seattle, WA 98170, USA}
\affiliation{Department of Chemistry, University of Chicago, Chicago, IL 60637, USA}%

\author{Yi-Ting Chen}
\affiliation{Amazon Braket, New York, NY 10018, USA}

\author{Brajesh Gupt}
\affiliation{AWS Worldwide Specialist Organization, Seattle, WA 98170, USA}

\author{Martin Suchara}
\email{M.S.' contributions to this work were made while employed by Amazon.}
\affiliation{Microsoft Azure Quantum, Redmond, WA 98052, USA}

\author{Anh Tran}
\affiliation{AWS Worldwide Specialist Organization, Seattle, WA 98170, USA}

\author{Sam McArdle}
\email{sammcard@amazon.com}
\affiliation{AWS Center for Quantum Computing, Pasadena, California, CA 91125, USA}%

\author{Giulia Galli}
\email{gagalli@uchicago.edu}
\affiliation{Department of Chemistry, University of Chicago, Chicago, IL 60637, USA}
\affiliation{Pritzker School of Molecular Engineering, University of Chicago, Chicago, IL 60637, USA}
\affiliation{Materials Science Division and Center for Molecular Engineering, Argonne National Laboratory, Lemont, IL 60439, USA}%

\date{\today}

\begin{abstract}
Solving the electronic structure problem of molecules and solids to high accuracy is a major challenge in quantum chemistry and condensed matter physics. The rapid emergence and development of quantum computers offer a promising route to systematically tackle this problem. Recent work by Huggins et al. [\textit{Nature}, \textbf{603}, 416 (2022)] proposed a hybrid quantum-classical quantum Monte Carlo (QC-QMC) algorithm using Clifford shadows to determine the ground state of a Fermionic Hamiltonian. This approach displayed inherent noise resilience and the potential for improved accuracy compared to its purely classical counterpart. Nevertheless, the use of Clifford shadows introduces an exponentially scaling post-processing cost. In this work, we investigate an improved QC-QMC scheme utilizing the recently developed Matchgate shadows technique [\textit{Commun. Math. Phys.}, \textbf{404}, 629 (2023)], which removes the aforementioned exponential bottleneck. We observe from experiments on quantum hardware that the use of Matchgate shadows in QC-QMC is inherently noise robust. We show that this noise resilience has a more subtle origin than in the case of Clifford shadows. Nevertheless, we find that classical post-processing, while asymptotically efficient, requires hours of runtime on thousands of classical CPUs for even the smallest chemical systems, presenting a major challenge to the scalability of the algorithm.
\end{abstract}

\maketitle

\section{Introduction} \label{introduction}

The ability to accurately solve the Schr\"odinger equation for interacting electrons will help tackle a multitude of problems in physics, chemistry and materials science, relevant to applications ranging from drug discovery~\cite{lin2020review} to the design of functional materials~\cite{king2012computational, hammes2021integration}. In the past century, multiple efforts have been devoted to solving the Schr\"odinger equation on classical computers, either by using suitable approximations and mean-field theories~\cite{kohn1996density} or by employing nearly exact methods such as full configuration interaction (FCI)~\cite{helgaker2013molecular} and quantum Monte Carlo (QMC)~\cite{foulkes2001quantum, austin2012quantum}. However, known algorithms for FCI scale exponentially on classical computers, while scalable solutions using QMC often suffer from the sign problem~\cite{troyer2005computational} and other numerical instabilities. In essence, highly entangled many-body wavefunctions of interacting electrons are hard to represent and optimize on classical computers.

The advent of quantum computers offers a promising route to tackle the interacting electron problem, either through quantum~\cite{aspuru2005simulated} or hybrid quantum-classical algorithms~\cite{peruzzo2014variational, motta2020determining}. One popular algorithm to compute Fermionic ground states is the variational quantum eigensolver (VQE) and its variants~\cite{misiewicz2023implementation, baek2023say, smart2022accelerated, grimsley2023adaptive}, where a parameterized quantum circuit is used to generate the wavefunction, and its optimization is off-loaded to classical hardware. Nevertheless, current noisy intermediate-scale quantum (NISQ) hardware is limited by the depth of quantum circuits that it can implement. Moreover, VQE faces optimization challenges resulting from ansatz-- or noise--induced barren plateaus in energy landscapes~\cite{wang2021noise}, as well as large measurement overheads.

Recently Huggins et al.~\cite{huggins2022unbiasing} proposed a hybrid quantum-classical algorithm for quantum Monte Carlo (QC-QMC) and applied it to study the dissociation of diatomic molecules on Google's Sycamore quantum processor. The Monte Carlo algorithm is driven by sampling from a trial state $|\Psi_T\rangle$ prepared on the quantum computer, with the aim of producing smaller bias and better accuracy than its purely classical counterpart. The algorithm uses classical shadows~\cite{huang2020predicting} with random Clifford circuits to avoid iterative communication between classical and quantum hardware, which is desirable on near-term quantum devices as it minimizes latency from quantum-classical communication. Surprisingly, an inherent noise resilience was observed~\cite{huggins2022unbiasing} --- spurring both academic and industrial interest in the technique~\cite{xu2023quantum, zhang2022quantum, montanaro2023accelerating, mazzola2023quantum, kanno2023quantum, amsler2023quantum, kiser2023classical}. Nevertheless, currently known techniques for classically post-processing the Clifford shadows scale exponentially with system size, providing a barrier to achieving quantum advantage.

The inefficiencies of Clifford shadows for QC-QMC were recently addressed in Ref.~\cite{wan2023matchgate} by replacing the Clifford circuits with Matchgate circuits~\cite{zhao2021fermionic}. To the best of our knowledge, the proposal has not yet been experimentally evaluated in the literature (although we note recent numerical investigations probing the statistical properties of Matchgate shadows applied to chemical systems~\cite{scheurer2023tailored,kiser2023classical}). In particular, it is unclear whether the Matchgate shadows approach exhibits the noise resilience that underpinned the success of the original (unscalable) Clifford shadows approach~\cite{huggins2022unbiasing}, as the theoretical justification for the noise resilience of Clifford shadows does not immediately extend to the Matchgate case. In addition, while the proposal is formally efficient, its practicality for studying realistic systems has yet to be established.

In this work, we numerically and experimentally study the QC-QMC algorithm, incorporating Matchgate shadows. In particular, we probe the noise resilience of the algorithm on quantum hardware. We observe and theoretically characterize the persistence of the noise resilience that was observed for Clifford shadows, uncovering a more subtle microscopic origin. We also developed improvements to existing robust classical shadow protocols~\cite{chen2021robust, zhao2023group, wu2023error} which can mitigate state preparation noise, and may be of independent interest. The Matchgate shadows approach reduces the classical post-processing cost from exponential to formally polynomial. However, its use in QC-QMC has high degree polynomial scaling in the system size, resulting from computing the local energy of each QMC walker. The algorithm also has large constant factors resulting from the number of samples used in the classical shadow post-processing, as well as the number of QMC walkers and the number of QMC timesteps. These final two factors are properties of the QMC algorithm itself, rather than aspects of the quantum algorithm. We observe hours of post-processing runtime for even the smallest chemical systems, using thousands of CPUs. This fundamentally challenges the scalability of the approach. We thus validate the persistence of one of the major strengths of QC-QMC, as well as highlight the challenges that must be overcome if the technique is to provide practical quantum advantage.

The rest of the paper is organized as follows. In Sec.~\ref{Preliminaries} we discuss how QC-QMC is performed and the basics of classical shadows, which serve as prerequisites for subsequent discussions. In Sec.~\ref{Results}, we first discuss prior work on noise robust classical shadows~\cite{chen2021robust,koh2022classical} (and their recently developed Matchgate variants~\cite{zhao2023group,wu2023error}) and give an extension to these methods that mitigate state preparation noise in QC-QMC. Second, we experimentally and analytically confirm the natural noise resilience of the Matchgate shadow protocol as used in QC-QMC. Finally, we apply QC-QMC to simulate the dissociation of hydrogen, and the ground state energy of a solid-state spin defect system, on different quantum hardware systems. These calculations reach agreement with the reference values even in the absence of error mitigation. In Sec.~\ref{Discussions}, we discuss the merits and shortcomings of QC-QMC and how it is related to other quantum-classical hybrid algorithms and purely classical QMC methods. Finally, Sec.~\ref{Conclusions} concludes our work with a summary.

\begin{figure*}[hbt!]
    \centering
    \includegraphics[width=\textwidth]{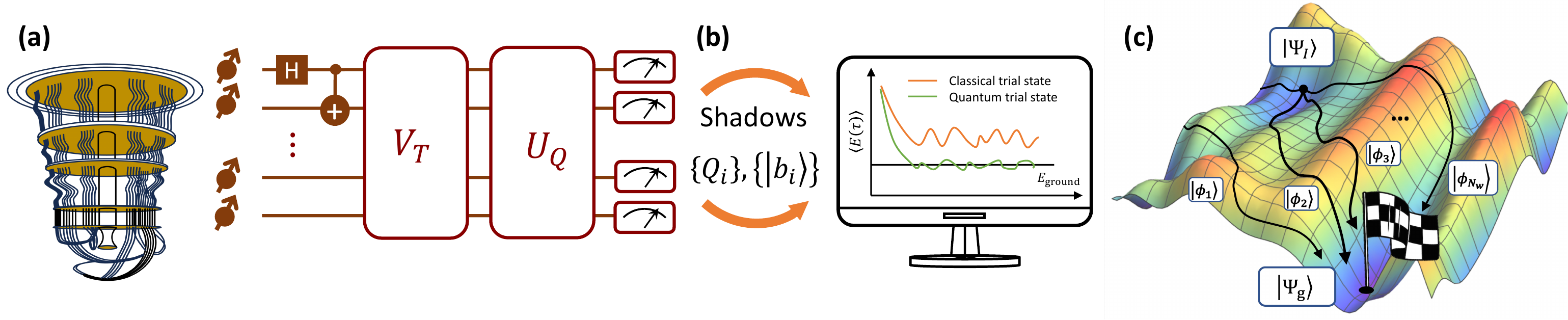}
    \caption{Workflow for the hybrid quantum-classical quantum Monte Carlo (QC-QMC) algorithm: (a). An equal superposition of the all-zero state $|\textbf{0}\rangle$, and the quantum trial state $|\Psi_T\rangle$ is prepared on a quantum computer, followed by the twirling of random unitary circuits $U_Q$ and measurements; (b). The measurement outcomes $\{|b_i\rangle\}$ and the random circuit unitary $U_Q$---which together constitute the information required to reconstruct classical shadows---are communicated to the classical computers; (c). The QMC procedure is carried out on a classical computer, where the walker states $|\phi_i\rangle$ evolve on the potential energy surface from the initial state $|\Psi_I\rangle$ towards the target ground state $|\Psi_g\rangle$.}
    \label{fig:workflow}
\end{figure*}

\section{Preliminaries} \label{Preliminaries}

The workflow of the QC-QMC method is shown in Fig.~\ref{fig:workflow}, where the computation on quantum and classical hardware is partitioned by classical shadows. In this section, we first provide a brief overview of the auxiliary-field quantum Monte Carlo (AFQMC) algorithm~\cite{motta2018ab} used in the QC-QMC method. Then we introduce the classical shadows formalism and discuss its use in QC-QMC. Readers familiar with this background material can skip to the summary of our results in Sec.~\ref{Results}.

\subsection{Auxiliary-field quantum Monte Carlo} \label{AFQMC}

Within the electronic structure community, QMC refers to a family of methods that bypass the explicit optimization of the many-body wavefunction in an exponentially large Hilbert space. Specifically, probabilistic sampling is applied in a subspace of the full Hilbert space, resulting in a polynomial scaling of the memory required for the evaluation of the ground state energy of a given system. The auxiliary-field quantum Monte Carlo (AFQMC) method is one example of projector QMC methods where a stochastic imaginary time evolution of the wavefunction is carried out by propagating samples on a manifold of nonorthogonal Slater determinants. Each sample $\{|\phi_l\rangle\}$ is usually called a “walker”. The ground state energy is estimated as
\begin{equation}
    E = \frac{\sum_l w_l e^{i\theta_l} E_l^{\text{loc}}}{\sum_l w_l e^{i\theta_l}},\;\;\;\; E_l^{\text{loc}} = \frac{\langle \Psi_T|H|\phi_l\rangle}{\langle\Psi_T|\phi_l\rangle}, \label{eq:gs_estimator}
\end{equation}
where $E_l^{\text{loc}}$ is the local energy of each walker, $w_l, \theta_l$ are its weight and phase, $H$ is the electronic structure Hamiltonian, and $|\Psi_T\rangle$ is a trial state defined as an approximation to the true ground state. The walker wavefunction and its weight and phase are updated every timestep according to functions of $\frac{\langle \Psi_T|\phi_i\rangle}{\langle \Psi_T|\phi_j\rangle}$ that perform the imaginary time evolution, see App.~\ref{app:AFQMC} for details.

For a generic Hamiltonian, however, similar to other projector QMC methods, AFQMC suffers from the sign problem, or more precisely, the phase problem where the phases $\theta_l$ of the walkers are evenly distributed in $[0, 2\pi)$~\cite{motta2018ab, zhang201315}. It leads to the ground state energy estimator, i.e., Eq.~\ref{eq:gs_estimator} (computed from the ensemble average of all the walkers) experiencing an exponentially fast decay of the signal-to-noise ratio and a large statistical error. To remedy the phase problem, a common solution is the so-called phaseless approximation (ph)~\cite{motta2018ab}, where $\theta_l = 0$ is enforced by a modified update rule, thus constraining the evolution of the walkers, see App.~\ref{app:AFQMC} for details. Such an approximation, while controlling the phase problem, introduces a bias to the estimated ground state energy. The magnitude of the bias is largely determined by how well the trial state $|\Psi_T\rangle$ represents the exact ground state. On classical computers, $|\Psi_T\rangle$ is usually chosen as either a single Slater determinant, e.g., the Hartree-Fock state, or a linear combination of Slater determinants, to ensure a polynomial scaling in computational time~\cite{mahajan2021taming}.

The insight of Ref.~\cite{huggins2022unbiasing} is that on a quantum computer, one may efficiently compute overlaps between Slater determinant walker states and a much wider range of trial states, e.g., the unitary coupled cluster (UCC) state~\cite{peruzzo2014variational, mcclean2016theory}. We refer to implementing AFQMC in this way as QC-AFQMC. Such states are potentially “closer” to the target ground state than those adopted as classical trials, and could lead to more accurate ground state energy estimations. It is currently unclear what are the best AFQMC trial states that can be prepared on a quantum computer, in low circuit depth. This is a similar issue to that of choosing a well-motivated and implementable ansatz circuit in VQE. We will discuss this problem in Sec.~\ref{Discussions}, and assume for now that a suitable AFQMC trial state can be efficiently prepared on a quantum computer. Then, the central issue for the QC-AFQMC scheme lies in how to evaluate the overlap amplitude $\langle\Psi_T|\phi\rangle$ for each walker, at each timestep. Efforts in the literature have branched out in two directions: Xu and Li~\cite{xu2023quantum} designed circuits based on the Hadamard test~\cite{ekert2002direct} to efficiently compute the overlap between the walker states and the trial state. On the other hand, Huggins et al.~\cite{huggins2022unbiasing} computed the amplitude by using the classical shadows technique applied to $|\Psi_T\rangle$. The former approach suffers a significant overhead in quantum computation due to the large number of walkers typically used in AFQMC. In addition, it requires iterative communication between the classical and quantum hardware, since $|\Psi_T\rangle$ is queried at every timestep during the time evolution. We therefore followed the approach of Huggins et al. in this work, and we provide a detailed scaling comparison of the two approaches in Section~\ref{Discussions}.

Finally, as originally pointed out by Ref.~\cite{huggins2022unbiasing} and then re-emphasized by Ref.~\cite{mazzola2022exponential, lee2022response}, the overlap amplitudes $\langle \Psi_T|\phi\rangle$ decay exponentially with system size in typical systems. This decay suggests that to reduce the relative errors of overlap amplitudes within acceptable thresholds, an exponentially large number of measurement shots would be necessary, resulting in the QC-QMC approach becoming non-scalable. This is supported by benchmarks on a transverse-field Ising model with varying sizes and measurement shots (on a noiseless simulator) using a Green's function Monte Carlo method~\cite{mazzola2022exponential, lee2022response} (as the QMC component). This issue is expected to plague both classical and quantum-assisted quantum Monte Carlo algorithms, and can be seen as a manifestation of the general Quantum Merlin-Arthur hardness of electronic structure calculations~\cite{huggins2022unbiasing}. Nevertheless, Ref.~\cite{lee2022response} explored strategies for postponing or alleviating this exponential bottleneck, including using more sophisticated walker wavefunctions (e.g. beyond a single Slater determinant) in QMC. It is an open question as to whether these strategies can be made compatible with classical shadows QC-QMC (note that these overlaps could alternatively be measured using the Hadamard test). If the vanishing overlaps issue can be mitigated in sufficiently large system sizes (e.g. $\sim 100$ orbitals), then it may be possible to achieve practical quantum advantage over classical QMC methods (and potentially other classical quantum chemistry algorithms) for systems of interest~\cite{huggins2022unbiasing}. In this study, we only consider small molecular systems which are not at the scale where the exponentially vanishing overlaps become a concern, opting to focus solely on the practical implementation challenges of QC-QMC. We leave investigating methods for testing and mitigating vanishing overlaps for future work.

\subsection{Classical shadows} \label{Shadow tomography}

As mentioned in the previous subsection, the need to evaluate $\langle\Psi_T|\phi\rangle$ for all walkers at each timestep is a key consideration for the practicality of the algorithm. In Ref.~\cite{huggins2022unbiasing} it was shown that the measurements required in QC-AFQMC can be recast into the following framework. Let $\rho$ denote the density matrix of an $n$-qubit quantum state that we know how to prepare, and $\{O_i\}$ denote a collection of $M$ observables whose expectation values, i.e., $\text{tr}(O_i\rho)$ we wish to estimate. Classical shadow tomography~\cite{huang2020predicting} provides a way to estimate these quantities with a cost that only scales logarithmically with $M$. Specifically, it estimates each $\text{tr}(O_i\rho)$ up to some error $\epsilon$ by the following procedure: i) choose a distribution $\mathcal{D}$ of unitary transformations; ii) sample random unitaries $U \in \mathcal{D}$ from the chosen distribution and iii) measure the state $U \rho U^{\dagger}$ in the computational basis $\{|b\rangle\}$ (we refer to the tensor products of $\{|0\rangle, |1\rangle\}$ of each qubit as the \textit{computational basis}, i.e., $\{|b\rangle\}_{b\in \{0, 1\}^n }$) to obtain the measurement outcome $|b\rangle\langle b|$. Consider the state $U^{\dagger} |b\rangle\langle b| U$; in expectation, the mapping from $\rho$ to $U^{\dagger} |b\rangle\langle b| U$ defines a quantum channel,
\begin{equation}
\begin{split}
    \mathcal{M}(\rho) & := \underset{U\sim \mathcal{D}}{\mathbb{E}} \left[U^{\dagger} |\hat{b}\rangle\langle \hat{b}| U\right]\\
    & = \underset{U\sim \mathcal{D}}{\mathbb{E}} \sum_{b\in \{0,1\}^n} \langle b|U\rho U^{\dagger}|b\rangle U^{\dagger} |b\rangle\langle b| U, \label{eq:shadow_channel}
\end{split}
\end{equation}
where $\mathbb{E}$ denotes the operation of averaging, and the hat represents a statistical estimator.

In the classical shadows framework, we require $\mathcal{M}$ to be invertible, which is true if and only if the collection of measurement operators defined by drawing $U \in \mathcal{D}$ and measuring in the computational basis is tomographically complete. Assuming that these conditions are satisfied, we can apply $\mathcal{M}^{-1}$ to both sides of above equation, yielding
\begin{equation}
\begin{split}
    \rho = \underset{U\sim \mathcal{D}}{\mathbb{E}}[\hat{\rho}] & = \mathcal{M}^{-1}\left(\underset{U\sim \mathcal{D}}{\mathbb{E}} \left[U^{\dagger} |\hat{b}\rangle\langle \hat{b}| U\right]\right)\\
    & = \underset{U\sim \mathcal{D}}{\mathbb{E}} \left[\mathcal{M}^{-1} \left(U^{\dagger} |\hat{b}\rangle\langle \hat{b}| U\right)\right].
\end{split}
\end{equation}
We call the collection $\{\mathcal{M}^{-1}(U^{\dagger} |\hat{b}\rangle \langle \hat{b}| U)\}$ the classical shadows of $\rho$. These shadows can be used to estimate the expectation values $\text{tr}(O_i \rho)$,
\begin{equation}
    \langle O_i\rangle = \underset{U\sim \mathcal{D}}{\mathbb{E}}\text{tr}\left[O_i \mathcal{M}^{-1} \left(U^{\dagger} |\hat{b}\rangle\langle \hat{b}| U\right)\right],
\end{equation}
each within error $\epsilon$ with probability at least $1-\delta$, with a number of samples that scales as:
\begin{equation}
    N_{\text{sample}} = \mathcal{O}\left(\frac{\log(M/\delta)}{\epsilon^2} \max_{1< i < M}\text{Var}[\hat{o}_i]\right). \label{eq:sample_complexity}
\end{equation}
In the above equation, we define $\hat{o}_i$ as an estimator of $\text{tr}(O_i \rho)$ and $\text{Var}[\hat{o}_i]$ represents the variance of estimator $\hat{o}_i$, which could be bounded by the shadow norm of $O_i$~\cite{huang2020predicting}. Importantly for QC-AFQMC, $N_{\text{sample}}$ only scales logarithmically with the number of target observables $M$\footnote{We note that the scaling in Eq.~\ref{eq:sample_complexity} was rigorously proven in Ref.~\cite{huang2020predicting} using a median-of-means estimator. Ref.~\cite{helsen2023thrifty, zhao2021fermionic} discussed when a mean estimator would suffice for Clifford and Matchgate shadows, respectively. We observed numerically that a mean estimator suffices for shadows in our small system QC-QMC experiments.}.

Formally, the condition that the measurement channel is invertible is sufficient for performing the classical shadows protocol. In practice, it is desirable that the protocol is efficient both in terms of quantum and classical resources. This means that there should be an efficient procedure to sample unitaries $U$ from $\mathcal{D}$ and implement them on a quantum computer, and in addition the variance of the estimates $\hat{o}_i$ is polynomial in system size. Moreover, it should be efficient to compute the expectation values with respect to the shadows, $\text{tr}\left[O_i \mathcal{M}^{-1}\left(U^{\dagger} |\hat{b}\rangle\langle \hat{b}| U\right)\right]$, on a classical computer.

In QC-AFQMC, the overlap amplitude $\langle\Psi_T|\phi\rangle$ is not a physical observable. Nevertheless, it can be measured within the framework of classical shadows by rewriting it as~\cite{huggins2022unbiasing}
\begin{equation}
\begin{split}
    \langle\Psi_T|\phi\rangle & = 2\text{tr}\left(|\phi\rangle\langle \mathbf{0}|\rho\right)\\
    & = 2 \underset{U\sim \mathcal{D}}{\mathbb{E}} \text{tr} \left[|\phi\rangle\langle \mathbf{0}| \mathcal{M}^{-1}\left(U^{\dagger} |\hat{b}\rangle\langle \hat{b}| U\right)\right],
\end{split}
\end{equation}
where $\rho$ is the density matrix corresponding to the state $\frac{1}{\sqrt{2}}\left(|\mathbf{0}\rangle + |\Psi_T\rangle\right)$ and $|\mathbf{0}\rangle = |0\rangle^{\otimes n}$. Here $|\phi\rangle\langle \mathbf{0}|$ plays the role of the operator $O$.

Ref.~\cite{huggins2022unbiasing} used random Clifford circuits to perform classical shadows, where the overlap amplitude is post-processed as
\begin{equation}
    \langle\Psi_T|\phi\rangle = 2 f^{-1} \underset{U\sim \text{Cl}(2^n)}{\mathbb{E}} \left[\langle \mathbf{0}|U^{\dagger}|\hat{b}\rangle \langle\hat{b}|U|\phi \rangle\right]. \label{eq:ovlp_Clifford}
\end{equation}
In the above equation, $f = (2^n + 1)^{-1}$ is the only (non-trivial) eigenvalue of $\mathcal{M}$, since the Clifford group has only one non-trivial irreducible representation (irrep). The first term $\langle \mathbf{0}|U^{\dagger}|b\rangle$ can be efficiently computed using the Gottesman-Knill theorem~\cite{aaronson2004improved}, but the second term $\langle b|U|\phi\rangle$, in general, cannot be efficiently estimated due to $|\phi\rangle$ being a random Slater determinant. As noted in Ref.~\cite{huggins2022unbiasing}, the method of Ref.~\cite{tang2019quantum} could be used to efficiently estimate this quantity up to an additive error. Nevertheless, the $f^{-1}$ prefactor exponentially amplifies this error, eliminating the efficiency of the proposed scheme.

Another interesting observation of Ref.~\cite{huggins2022unbiasing} is that the evaluation of overlap ratios, which drives the AFQMC algorithm, is resilient to quantum hardware noise. This phenomenon can be understood by first noting that noise will change the prefactor $f$ to a modified prefactor $\widetilde{f}$ (as explained in detail in Sec.~\ref{subsection:robust_shadow}). This prefactor cancels out when taking the ratio, recovering the noiseless result without requiring any additional error mitigation. We refer to this as having an inherent noise resilience.

Refs.~\cite{zhao2021fermionic,wan2023matchgate,  low2022classicalshadows, ogorman2022fermionic} investigated the replacement of the random Clifford ensemble with a random Matchgate ensemble (in some cases restricting to Clifford Matchgates and/or number conserving Matchgates). The Matchgate circuit is the qubit representation of the Fermionic Gaussian transformation assuming the Jordan-Wigner (JW) transformation~\cite{wigner1928paulische} is used. We therefore adopt the JW mapping throughout this work. Unlike the Clifford group, the Matchgate group has $(n+1)$ (even) irreps, see App.~\ref{app:Matchgates} for details. We focus on the results of Ref.~\cite{wan2023matchgate}, which 
\begin{itemize}
    \item Proved the equivalence between the continuous and discrete Matchgate ensembles for classical shadows.
    \item Explicitly showed how to efficiently compute the overlap between trial states and Slater determinant states using Matchgate shadows, including bounding the worst case variance in the estimate.
\end{itemize}
We refer the interested reader to Ref.~\cite[Sec. III D]{wan2023matchgate} for a more detailed comparison between these recent works. The results of Ref.~\cite{wan2023matchgate} overcame the post-processing limitations of Clifford shadows, showing that overlap amplitude can be computed as
\begin{equation}
    \langle\Psi_T|\phi\rangle = 2\sum_{l=0}^n f^{-1}_{2l} \underset{Q \sim B(2n)}{\mathbb{E}} \text{tr} \left[|\phi\rangle\langle 0| \Pi_{2l} \left(U_Q^{\dagger} |\hat{b}\rangle\langle \hat{b}| U_Q\right)\right], \label{eq:ovlp_Matchgate}
\end{equation}
where the random Matchgate circuits $U_Q$ are prepared by sampling random signed permutation matrices $Q$ ($Q \in B(2n)$, where $B$ represents the Borel group). In Eq.~\ref{eq:ovlp_Matchgate}, $\Pi_{2l}$ is the projector associated with the $l$-th even irrep, and its eigenvalue $f_{2l}$ is
\begin{equation}
    f_{2l} = \begin{pmatrix} 2n\\ 2l\end{pmatrix}^{-1} \begin{pmatrix} n\\ l\end{pmatrix}.
\end{equation}
Importantly, Eq.~\ref{eq:ovlp_Matchgate} can be efficiently solved using the matrix Pfaffian, with a scaling of $\mathcal{O}\left((n-\zeta/2)^4\right)$ by polynomial interpolation~\cite{wan2023matchgate}, where $\zeta$ denotes the number of electrons. As such, the exponential post-processing bottleneck present in the Clifford shadows approach is eliminated through the use of Matchgate shadows.

Inherent noise resilience, analogous to that observed for overlap ratios computed via Clifford shadows, has not yet been established for overlap ratios computed via Matchgate shadows. One sufficient but not necessary condition, (following from the Clifford case), would be $\widetilde{f}_{2l} = \alpha f_{2l}, \;l \in \{0, 1, \dots, n\}$, where $\alpha$ is a proportionality constant. In Sec.~\ref{subsection:noise_resilience} we investigate the noise resilience of overlap ratios computed via Matchgate shadows. While our experimental results do not display a universal proportionality constant between each $\widetilde{f}_{2l}$ and $f_{2l}$, we nevertheless find that the computed ratios are inherently robust to noise, and provide justification for this observation.

\section{Results} \label{Results}

In this section, we present our results, divided into three parts. We first introduce prior work on noise robust classical shadows, and provide an improvement that mitigates state preparation noise in QC-AFQMC. Second, we experimentally test for noise resilience of both overlap amplitudes and ratios of overlap amplitudes, evaluated via the Matchgate shadow technique. Notably, we observe that the ratios are resilient to the effects of noise (as was previously observed for Clifford shadows~\cite{huggins2022unbiasing}). We provide a theoretical explanation for the observed noise resilience and its limitations, which has more subtle origins than in the Clifford case. Third, we apply the QC-AFQMC algorithm to compute i) the dissociation curve of the hydrogen molecule on a superconducting qubit quantum computer; and ii) the ground state of a negatively charged nitrogen-vacancy (NV) center in diamond on a trapped ion quantum computer. 

We use the following notation in this section: we denote overlap amplitudes obtained from noiseless shadows with an ordinary bracket, $\langle\Psi_T|\phi\rangle$, quantities subject to quantum noise with a tilde bracket, $\widetilde{\langle\Psi_T|\phi\rangle}$, and we use a tilde plus subscript $r$ to denote quantities obtained from a robust shadow protocol, $\widetilde{\langle\Psi_T|\phi\rangle}_r$.

\subsection{Robust Matchgate shadow protocol} \label{subsection:robust_shadow}

Refs.~\cite{chen2021robust,koh2022classical} developed a “robust shadow protocol”, showing that when the quantum noise is gate-independent, time-stationary, and Markovian (GTM), and the state preparation of $\rho$ is perfect, a noise-free expectation value can be obtained. The protocol works by replacing the eigenvalue(s) $f$ of $\mathcal{M}$ with modified value(s) $\widetilde{f}$ that compensates for the effect of noise. The robust scheme uses additional shadow-like circuits to estimate the value(s) of $\widetilde{f}$. It can be viewed as passive error mitigation, as the updates are made purely in the classical post-processing of the shadow expectation value. The framework was originally applied to global and local Clifford shadows. Recently, two separate works extended the robust shadow protocol to the Matchgate setting~\cite{zhao2023group, wu2023error}. The latter of these works observed a close connection between the circuit used for the robust Matchgate shadow protocol and prior work on Matchgate randomized benchmarking~\cite{helsen2022matchgate}. We will closely follow this approach, presenting derivations for completeness in App.~\ref{app:robust_shadow}.

The robust Matchgate shadows protocol~\cite{wu2023error} estimates $\{\widetilde{f}_{2l}\}$ using the following steps: i) prepare the all-zero state $|\textbf{0}\rangle$ on a quantum computer; ii) sample $Q \in B(2n)$ and apply the Matchgate circuit $U_Q$ to $|\textbf{0}\rangle$; iii) measure in the computational basis and collect measurement outcomes $|b\rangle$. The expressions used to calculate $\{\widetilde{f}_{2l}\}$ from these measurement outcomes are given in App.~\ref{app:robust_shadow}. The sample complexity is asymptotically equivalent to that used in the noiseless case~\cite{wu2023error}.

The established robust shadow scheme focuses on noise in the shadow unitary ($U_Q$ in Fig.~\ref{fig:workflow}a), and does not mitigate any noise occurring during state preparation ($V_T$ in Fig.~\ref{fig:workflow}a). We extend the method to (partly) account for noise that occurs during the state preparation unitary. In QC-AFQMC, we evaluate the overlap amplitudes by measuring classical shadows of an equal superposition state $\frac{1}{\sqrt{2}}(|\Psi_T\rangle + |\mathbf{0}\rangle)$. This state is prepared via a circuit that first generates a superposition of the Hartree-Fock state with $\ket{\mathbf{0}}$ state, and then applies a unitary $V_T$, such that $V_T|\Psi_{\text{HF}}\rangle = |\Psi_T\rangle$ and $V_T|\mathbf{0}\rangle = |\mathbf{0}\rangle$~\cite{obrien2021VerifiedQPE}. As discussed above, to estimate the values of $ \{\widetilde{f}_{2l}\}$ required in the robust Matchgate protocol, we apply a Matchgate shadow circuit to the initial state $\ket{\mathbf{0}}$. We can partially account for the noise in the state preparation circuit in the following way. We minimally alter the QC-AFQMC circuit in Fig.~\ref{fig:workflow}a, such that it measures $\{\widetilde{f}_{2l}\}$, by removing the initial Hadamard gate. In a noiseless setting, this modified version of $V_T$ would still result in $\rho = \ket{\mathbf{0}}\bra{\mathbf{0}}$, as required. In the noisy setting, we can treat the noise as originating in the Matchgate circuit, which is then mitigated by the robust protocol. While we do not yet have a full mathematical characterization of the noise models that can be mitigated by our enhanced robust Matchgate shadows technique, we demonstrate the efficacy of this approach through numerical simulations and experiments on quantum hardware, presented in the following subsection.

\subsection{Noise resilience of overlap amplitudes and their ratios}

As discussed previously, all of the key quantities used in the AFQMC algorithm, e.g., the local energy, can be expressed as linear combinations of the ratio of overlaps $\frac{\langle\Psi_T|\phi_i\rangle}{\langle\Psi_T|\phi_j\rangle}$. Hence we can view the propagation in imaginary time as being driven by these overlap ratios. We can estimate the individual overlaps from Matchgate shadows. In this section, we experimentally probe the impact of noise on both the overlap amplitudes, and their ratios, by measuring these quantities via Matchgate shadows implemented on the IBM Hanoi superconducting qubit quantum computer. As a test system, we consider the hydrogen molecule in its minimal STO-3G basis, mapped to four qubits via the Jordan-Wigner transform. We use a trial state of the form $\ket{\Psi_T} = \alpha \ket{1100} + \beta \ket{0011}$, which is a linear combination of the Hartree-Fock configuration and double excited state, respectively.

\subsubsection{Evaluating the overlap amplitudes}
We compute the overlaps $\langle\Psi_T|\phi_i\rangle$ between the trial state and 16 randomly sampled Slater determinants, using Eq.~\ref{eq:ovlp_Matchgate}, as a function of the number of Matchgate shadows used. We consider both noisy Matchgate shadows, and those corrected using the robust Matchgate shadow protocol introduced in Sec.~\ref{subsection:robust_shadow}.

In Fig.~\ref{fig:MAE_shadow} (upper panel) we show the mean absolute error (MAE) of the noisy and noise-robust overlaps, with respect to the ideal noiseless value. We observe that the noisy results differ significantly from the true values, and their accuracy is limited by the effects of noise. The two robust shadow approaches (accounting for state preparation error (SP) and not accounting for it) mitigate the impact of hardware noise. In particular, compensating for the effects of state preparation noise results in a much smaller deviation from the ideal noiseless values than the standard robust approach. This observation signifies the importance of considering the state preparation noise when using robust shadows on quantum hardware.

\begin{figure}[hbt!]
    \centering
    \includegraphics[width=0.48\textwidth]{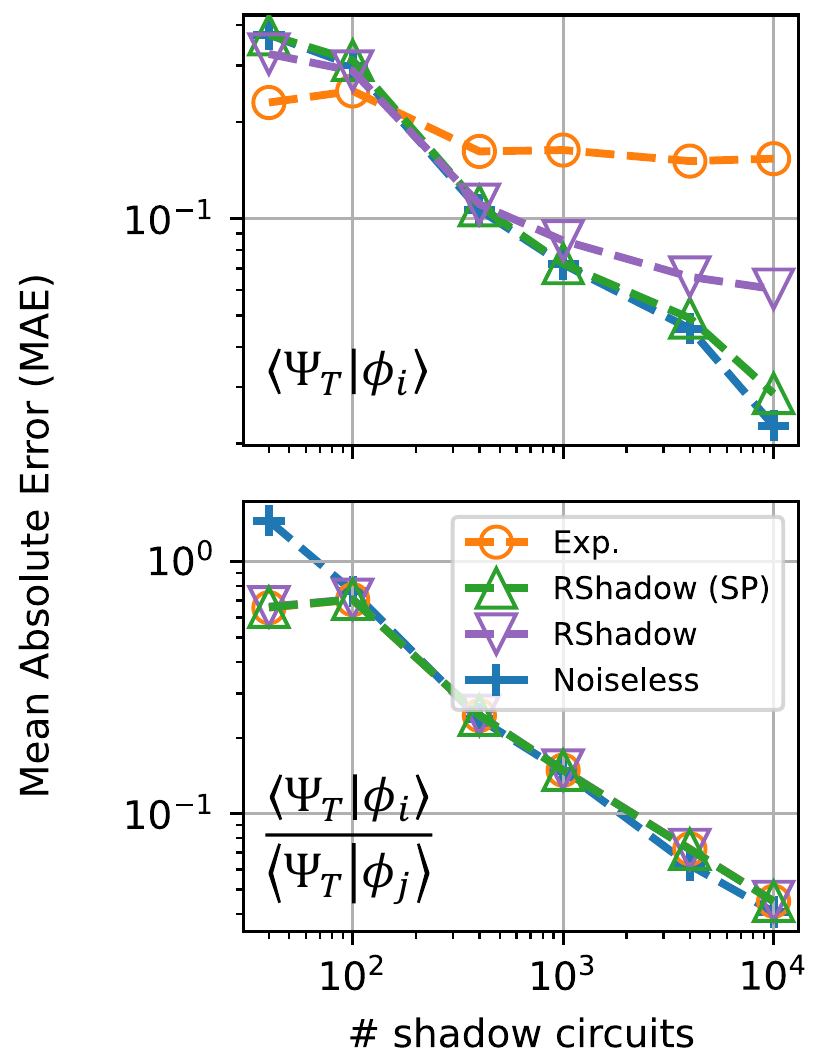}
    \caption{The mean absolute error (MAE) of overlap amplitudes $\langle\Psi_T|\phi_i\rangle$ and overlap ratios $\frac{\langle\Psi_T|\phi_i\rangle}{\langle\Psi_T|\phi_j\rangle}$ w.r.t. the number of Matchgate circuits, varying from 40 to 10,000. A total of 16,000 Matchgate shadow circuits are used in this experiment, each using 1024 measurement shots~\cite{zhou2023performance}. The raw experimental results (“Exp.”) and noiseless simulation (“Noiseless”) are plotted in orange circles and blue crosses, respectively. For the two robust protocols (“RShadow” and “RShadow (SP)”), we allocated another 16,000 Matchgate circuits for each to determine the coefficients $\widetilde{f}_{2l}$, thus doubling the measurement cost.}
    \label{fig:MAE_shadow}
\end{figure}

\subsubsection{Noise resilience in the evaluation of overlap ratios} \label{subsection:noise_resilience}

Using the results presented in the previous section for $\{\langle\Psi_T|\phi_i\rangle\}$, we can compute the 120 possible ratios of overlaps originating from the 16 Slater determinants sampled. In Fig.~\ref{fig:MAE_shadow} (lower panel) we show the MAE of these overlap ratios with respect to the ideal values (and verify in Fig.~\ref{fig:MAE_ratio_each} in App.~\ref{app:quantum_simulations} that this behavior holds for not only the MAE, but the errors associated with each overlap ratio). We observe the following behavior
\begin{equation}
    \frac{\langle\Psi_T|\phi_i\rangle}{\langle\Psi_T|\phi_j\rangle} \approx \frac{\widetilde{\langle\Psi_T|\phi_i\rangle}_r}{\widetilde{\langle\Psi_T|\phi_j\rangle}_r} \approx \frac{\widetilde{\langle\Psi_T|\phi_i\rangle}}{\widetilde{\langle\Psi_T|\phi_j\rangle}}, \label{eq:noise_resilience}
\end{equation}
when the number of shadow circuits exceeds 100. This is an indication of noise resilience in evaluating the overlap ratios, suggesting that the robust shadow protocol may not be required for QC-AFQMC. We provide theoretical justification for these observations below.

The first (approximate) equality in Eq.~\ref{eq:noise_resilience} will not hold for all noise models, as the robust Matchgate shadow protocol is only guaranteed to correct the effects of noise if its assumptions (GTM noise, noiseless state preparation) are satisfied. Nevertheless, for cases where the robust Matchgate protocol is able to correct the effects of noise the equality of the ratios follows directly. In our experiments, we found that the robust Matchgate protocol was able to correct the individual overlaps up to a small residual error attributed to residual state preparation noise. This leads to the approximate equality between the ratios.

The second (approximate) equality shows that the ratio of any two uncorrected overlap amplitudes (which are themselves inaccurate, see Fig.~\ref{fig:MAE_shadow} (upper panel)) yields the same results as ratios of overlaps computed via the robust Matchgate protocol. To understand this result, we first restate Eq.~\ref{eq:ovlp_Matchgate} for computing the overlap in the absence of noise
\begin{equation*}
    \langle\Psi_T|\phi\rangle = 2\sum_{l=0}^n f^{-1}_{2l} \underset{Q \sim B(2n)}{\mathbb{E}} \text{tr} \left[|\phi\rangle\langle \mathbf{0}| \Pi_{2l} \left(U_Q^{\dagger} |\hat{b}\rangle\langle \hat{b}| U_Q\right)\right].
\end{equation*}
We make use of the following Theorem:
\begin{restatable}{theorem}{MainTheorem}\label{Theorem:MainTheorem}
For an $n$ qubit state $\ket{\phi} = \sum_i c_i \ket{i}$ which is a linear combination of computational basis states $\ket{i}$ with fixed Hamming weight $\zeta$, 
\begin{equation}
    \mathcal{P}_{0, \zeta}\big[\Pi_{2l} \left(|\phi\rangle\langle \mathbf{0}|\right)\big] = b_{2l} |\phi\rangle\langle \mathbf{0}|,
\end{equation}
with
\begin{equation}
    b_{2l}=\begin{cases}
	2^{\zeta - n} \binom{n-\zeta}{l-\frac{\zeta}{2}}, & \text{if $\frac{\zeta}{2} \leq l \leq n - \frac{\zeta}{2}$}\\
        0, & \text{otherwise}
    \end{cases}
\end{equation}
such that $\sum_{l=0}^n b_{2l} = 1 $. Here $\mathcal{P}_{0, \zeta}$ denotes the projector onto the subspace spanned by the Hamming weight zero state $\ket{\mathbf{0}}$ and computational basis states of Hamming weight $\zeta$.
\end{restatable}
\begin{proof}
A full proof is given in App.~\ref{app:eigen-operator}, and we sketch a brief outline here. The operation $\Pi_{2l} \left(|\phi\rangle\langle \mathbf{0}|\right)$ can be re-expressed as $2^{-n} \sum_{j_l} \langle \mathbf{0}|\gamma_{j_l}^{\dagger} |\phi\rangle \gamma_{j_l}$ where $\gamma_{j_l}$ are strings of $2l$ Majorana operators (defined in App.~\ref{app:Matchgates}). The action of the Majorana strings is to create ``1's" in the $\ket{\mathbf{0}}$ state. Only terms with $\zeta$ 1's created contribute to the sum. Projecting the resulting sum into the restricted subspace ensures the `eigen-operator' property. The prefactor is given by counting arguments for the number of Majorana strings that can yield $\zeta$ 1's for a given $l$ value.
\end{proof}
The justification for restricting to the $\{0,\zeta\}$-electron subspace is that $\rho$ is an equal superposition of $|\mathbf{0}\rangle$ and $|\Psi_T\rangle$, and therefore the trace of $\rho$ with operators outside this subspace will necessarily be zero. This will still hold approximately true when $\rho$ is reconstructed from a finite number of classical shadows. This operator only appears in the classical post-processing step and is therefore not affected by noise. Therefore, the noisy overlap amplitude can be expanded as
\begin{equation}\label{Eq:NonRobustMatchgateForRatios}
    \widetilde{\langle\Psi_T|\phi\rangle} \approx 2\left(\sum_{l=0}^n f^{-1}_{2l} b_{2l}\right) \underset{Q \sim B(2n)}{\mathbb{E}} \left[\Big\langle \mathbf{0}\Big|U_Q^{\dagger}\Big| \hat{\tilde{b}}\Big\rangle \Big\langle \hat{\tilde{b}}\Big|U_Q\Big|\phi\Big\rangle\right],
\end{equation}
where in deriving this equation, we have used $\text{tr}\left(A \mathcal{M}^{-1}(B)\right) = \text{tr}\left(\mathcal{M}^{-1}(A) B\right)$, where operators $A, B$ are in the subspace where $\mathcal{M}$ is invertible~\cite{wan2023matchgate}. This expression can be compared to Eq.~\ref{eq:ovlp_Clifford} for Clifford shadows, noting the common feature of the separation of the expression into a sum over the channel index $l$, and the expectation over the shadow unitaries acting on the walker state $|\phi\rangle$. Eq.~\ref{Eq:NonRobustMatchgateForRatios} can also be understood as the way that the Matchgate channel is twirling the noise in each subspace, in a similar manner to that of the Clifford shadows. In order to mitigate the effects of noise with the robust Matchgate shadows approach, we could replace $f^{-1}_{2l}$ with $\widetilde{f}^{-1}_{2l}$
\begin{equation}\label{Eq:RobustMatchgateForRatios}
    \widetilde{\langle\Psi_T|\phi\rangle}_r \approx 2\left(\sum_{l=0}^n \widetilde{f}^{-1}_{2l} b_{2l}\right) \underset{Q \sim B(2n)}{\mathbb{E}} \left[\Big\langle \mathbf{0}\Big|U_Q^{\dagger}\Big| \hat{\tilde{b}}\Big\rangle \Big\langle \hat{\tilde{b}}\Big|U_Q\Big|\phi\Big\rangle\right].
\end{equation}
If the noise satisfies GTM assumptions (and state preparation is noise-free), Eq.~\ref{Eq:RobustMatchgateForRatios} is able to correct the effect of noise that manifests from the deviation of $\Big|\hat{\tilde{b}}\Big\rangle \Big\langle \hat{\tilde{b}}\Big|$ from $\Big|\hat{b}\Big\rangle \Big\langle \hat{b}\Big|$. If these assumptions on the noise are not satisfied, robust Matchgate shadows may not be able to perfectly recover the noiseless value. We can now observe that regardless of whether the robust procedure is used, the terms corresponding to the sum over $l$ in Eq.~\ref{Eq:NonRobustMatchgateForRatios} and Eq.~\ref{Eq:RobustMatchgateForRatios} are both independent of $|\phi\rangle$. As such, when computing the ratios $\frac{\widetilde{\langle\Psi_T|\phi_i\rangle}_r}{\widetilde{\langle\Psi_T|\phi_j\rangle}_r}$ and $\frac{\widetilde{\langle\Psi_T|\phi_i\rangle}}{\widetilde{\langle\Psi_T|\phi_j\rangle}}$, the prefactor will cancel between the numerator and denominator, leaving behind a ratio of estimators $\underset{Q \sim B(2n)}{\mathbb{E}}\Big[\dots\Big]$ that is identical for the robust case and the non-robust case. As such, if the ratio of overlaps evaluated via robust Matchgate shadows is able to approximately recover the noiseless results (as observed in our data), we conclude that the ratio of overlaps evaluated via regular Matchgate shadows is also able to approximately recover the noiseless value. This explains the data shown in Fig.~\ref{fig:MAE_shadow} (lower panel).

We emphasize that the cancellation of the pre-factor found here, i.e., $\sum_{l=0}^n \widetilde{f}^{-1}_{2l} b_{2l}$, should not be mistaken with that of the Clifford case, although they have a very similar form. This prefactor of Matchgate shadows, as a sum, contains the noise information of all $(n+1)$ even subspaces of the Matchgate channel. From the Clifford case, one might naturally think that the cancellation for the Matchgate shadow is due to $\widetilde{f}_{2l} = \alpha f_{2l}$ where $\alpha$ is independent of $l$. We verified that this is not the case by experimentally determining $\widetilde{f}_{2l}$. We observed that $\alpha_l$ could differ by up to 20\% for different $l$. These data are recorded in App.~\ref{app:quantum_simulations}. Therefore, this cancellation should be credited to the structure of $|\phi\rangle\langle \mathbf{0}|$, namely being an eigen-operator of $\Pi_{2l}$ in the restricted subspace, as proven in Theorem~\ref{Theorem:MainTheorem}. If another observable, without this property, were chosen, any observed noise resilience would have an alternative origin. For example, Ref.~\cite{scheurer2023tailored} shows a similar contractive property in evaluating the ratios, but only for noise models that ensure $\langle \textbf{0}|\rho_{\text{noise}}|\phi\rangle$ is zero, where $\rho_{\text{noise}}$ is the noise corrupted part of the density matrix.

\begin{figure}[hbt!]
    \centering
    \includegraphics[width=0.48\textwidth]{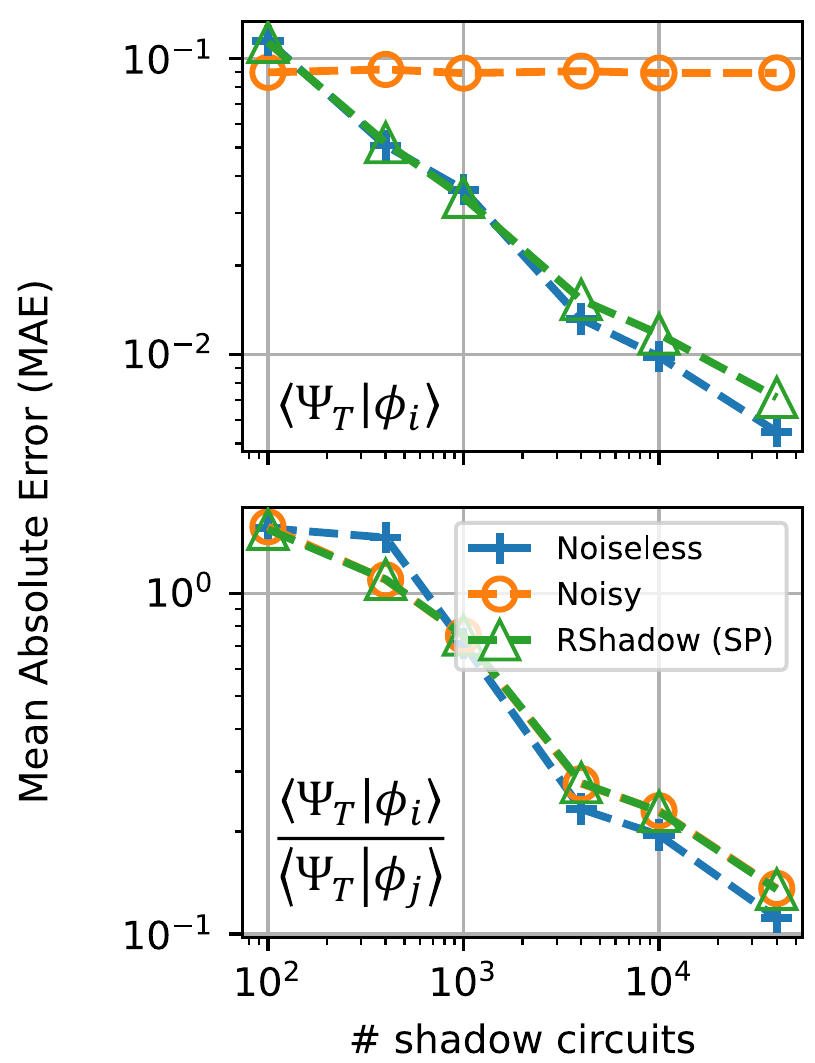}
    \caption{The MAE of estimating overlap amplitudes and overlap ratios w.r.t. the number of Matchgate shadow circuits for the water molecule (8 qubits). These results were obtained through numerical simulations under single-qubit asymmetric Pauli noise, where the Pauli-$X$, $Z$ and $Y$ error rates are set to 1, 3, 2\%, respectively.}
    \label{fig:8Q_convergence}
\end{figure}

As discussed above, robust shadows have been shown to correct for noise in the shadow circuit, assuming noiseless state preparation~\cite{chen2021robust}, or state preparation affected by global depolarizing noise~\cite{huggins2022unbiasing,scheurer2023tailored}. In a realistic setting, the state preparation step may suffer from more complicated noise which may prevent the technique from recovering the noiseless result. We performed numerical simulations to compute overlap ratios for the water molecule in a (4e, 4o) active space (8 qubits). We employed asymmetric Pauli errors as the noise model in this simulation (including in the state preparation circuit), and we refer to App.~\ref{app:quantum_simulations} for more details. The results, shown in Fig.~\ref{fig:8Q_convergence} corroborate the analysis provided above. Specifically, we observe from the upper panel that when computing $\widetilde{\braket{\Psi_T}{\phi_i}}$ the results differ significantly from the noiseless value. In contrast, when computing $\widetilde{\braket{\Psi_T}{\phi_i}}_r$ using the robust Matchgate protocol (with compensation for state preparation error) we are able to almost recover the noise-free value. We attribute the deviation from the noiseless value to the presence of some residual state preparation error, which violates the assumptions of the robust shadow protocol. In contrast, from the lower panel we observe that the ratios of overlaps evaluated via the robust Matchgate shadow protocol are no more accurate than those evaluated using the regular Matchgate shadow protocol. Moreover, both of these values deviate from the noiseless value by a small error due to residual (uncorrectable) state preparation noise. In the case of coherent errors during state preparation, we effectively prepare $\ket{\widetilde{\Psi}_T} = U \ket{\Psi_T}$ for some unitary $U$. As such, while we would still observe that $\frac{\widetilde{\langle\Psi_T|\phi_i\rangle}}{\widetilde{\langle\Psi_T|\phi_j\rangle}} \approx \frac{\widetilde{\langle\Psi_T|\phi_i\rangle}_r}{\widetilde{\langle\Psi_T|\phi_j\rangle}_r}$, neither ratio would likely recover the noiseless value.

Overall, the results of this subsection highlight a limitation of the robust Matchgate protocol, and its use in QC-AFQMC. Nevertheless, these results should be seen as an attractive feature of QC-AFQMC, as they suggest that the algorithm, driven by overlap ratios, is naturally resilient to the impact of many types of noise.

\subsection{Computing ground states with QC-AFQMC}

Having established the inherent noise resilience of the overlap ratios that drive the AFQMC algorithm, we demonstrate that this property contributes to noise resilience in the overall algorithm. We first compute the dissociation curve of the four-qubit hydrogen molecule studied in the previous subsections using the QC-AFQMC algorithm implemented on the IBM Hanoi superconducting qubit quantum processor. We then apply the QC-AFQMC algorithm to compute the ground state of an NV center in diamond, using the Aria trapped ion quantum processor from IonQ. In both settings, we emulate practical applications of QC-AFQMC by assuming a trial state generated by an imperfect VQE calculation. In both cases, we use a tailored UCCSD ansatz, see App.~\ref{app:quantum_simulations}, which was optimized until the energy was approximately $30\sim 80$ mHa higher than the FCI reference. We emphasize that these quantum hardware calculations do not seek to demonstrate any advantage over the classical QMC algorithm, as using an unrestricted Hartree-Fock or a multi-Slater trial state in classical AFQMC would also guarantee accurate ground state energies for these small system sizes.

\subsubsection{Hydrogen molecule dissociation}

In Fig.~\ref{fig:H2_AFQMC} (upper panel) we show the ground state energies obtained from QC-AFQMC at 5 different molecular geometries. If we target an error $\epsilon=10^{-2}$ on any overlap, with $\geq 99.99\%$ confidence, the rigorous sample complexity bounds presented in Eq.~\ref{eq:sample_complexity} imply the need to use $\sim1.1\times 10^5$ shadow circuits (with one measurement shot per circuit). As shown in Fig.~\ref{fig:MAE_shadow}, it was sufficient to use only 16000 circuits for each bond length, suggesting the bounds of Eq.~\ref{eq:sample_complexity} are overly pessimistic for the small system studied here. Given the noise resilience observed in evaluating the overlap ratios, the use of robust shadows was not necessary. In our AFQMC calculations, we used 4800 walkers for all numerical data presented here to ensure a negligible variance, and a timestep $\Delta \tau = 0.005\; \text{H.a.}^{-1}$ resulting in a negligible Trotter error. We chose the initial walker state to be the Hartree-Fock state, for each of the five distances studied here. The calculation was parallelized across 4800 CPU cores (1 core per AFQMC walker), as discussed in more detail below.

In the interest of conserving computational resources, only two of the five data points (solid points) were obtained using the scalable Matchgate shadows approach outlined above. The remaining three data points were obtained using an exponentially scaling approach used in Ref.~\cite{huggins2022unbiasing} and detailed in App.~\ref{app:classical_pp}, that is ultimately more efficient than the scalable Matchgate approach for small system sizes. We verified that the two schemes give the same results for overlap amplitudes. Thus, the two Matchgate-obtained points are used as a complete evaluation of the practicality of the algorithm, while the remaining three points are only presented to confirm the accuracy of the algorithm in the presence of hardware noise.

We observe that all five QC-AFQMC calculations converge to within computational accuracy of the FCI reference value, even without the use of the robust shadow protocol. In Fig.~\ref{fig:H2_AFQMC} (lower panel) we show the variation in energy as a function of imaginary time for the 0.75 \text{\r{A}} geometry. There is little observable difference ($\sim 0.1$ mHa) between the noisy experiment and its noiseless classical emulation.

We further note that even when the state preparation suffers from coherent errors, the QC-AFQMC algorithm is still able to recover accurate ground state energies, providing the effective quantum trial state $|\widetilde{\Psi}_T\rangle$, reconstructed from the noisy shadows, does not appreciably deviate from $|\Psi_T\rangle$. This is a result of the dissipative nature of the imaginary time process emulated by AFQMC, which will drive any initial state towards the ground state. This further enhances the inherent noise resilience of the QC-AFQMC algorithm.

\begin{figure}[hbt!]
    \centering
    \includegraphics[width=0.48\textwidth]{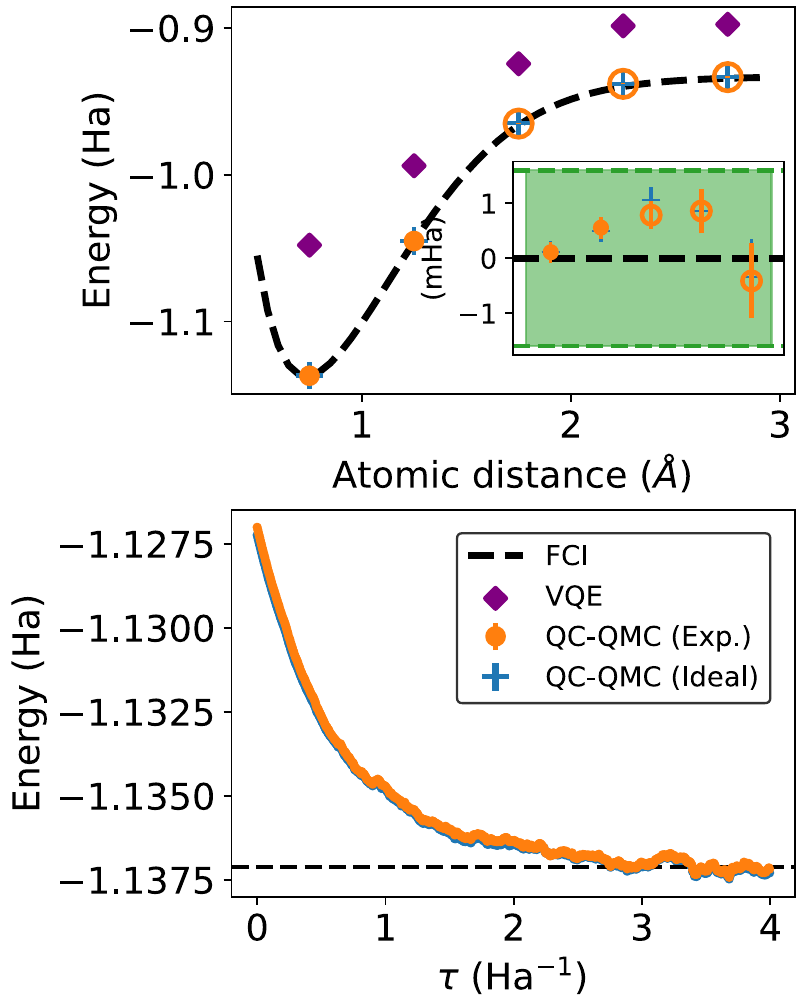}
    \caption{Estimations of the ground state energy of the hydrogen molecule at 5 different bond distances (upper). The QC-AFQMC calculation run on the IBM Hanoi quantum computer (solid and hollow orange circles labeled “exp.”) converges within computational accuracy for all distances (see inset), with a difference of $\sim 10^{-4}$ H.a. from its simulated noiseless counterpart (blue crosses). As discussed in the main text, filled circles denote results obtained with the scalable Matchgate shadows approach, while the empty circles denote results obtained with the non-scalable approach outlined in App.~\ref{app:classical_pp}. The quantum trial state is obtained from a noisy VQE simulation (purple diamond). An imaginary time evolution at 0.75 \text{\r{A}} is plotted in the lower panel, where the auxiliary fields sampled at each timestep are synchronized between the raw noisy experiment and noiseless reference.}
    \label{fig:H2_AFQMC}
\end{figure}

We emphasize that even for the hydrogen molecule in a minimal basis set, executing the QC-AFQMC algorithm using Matchgate shadows requires a substantial amount of resources, with classical computation dominating over the quantum subroutine. It was necessary to perfectly parallelize the classical AFQMC propagation, with one walker per CPU core\footnote{Intel Xeon Platinum Processor.} (this is only possible when not using population control\footnote{In QMC, population control refers to a technique to retain sampling efficiency by adjusting the walkers so that the algorithm avoids spending a disproportionate amount of time keeping track of walkers that contribute little to the energy estimate~\cite{motta2018ab}. It makes the algorithm more robust against sampling noise but necessitates communication between all the walkers and therefore renders the parallelization less effective.}). For each walker, it took approximately 1 minute to post-process the shadows from 16000 circuits in each timestep. As each step of the AFQMC algorithm must be performed sequentially, it could take up to hours or even days to complete the full evolution, see App.~\ref{app:quantum_simulations} for details. This total runtime partly comes from the large number of shadows, which plays the role of a pre-factor and could be further improved by implementing additional parallelization among the shadows, for example, processing the shadows corresponding to a single walker in parallel across a number of cores, rather than sequentially on the same core. The total runtime also stems from a scaling as high as $\mathcal{O}(n^{8})$ to evaluate the local energy of each walker, which is discussed in detail in App.~\ref{app:local_energy}. This high cost presents a practical concern for scaling this algorithm to larger systems of practical interest.

\subsubsection{Nitrogen-vacancy center ground state}

We applied the QC-AFQMC algorithm to study the electronic states of a point defect in a solid. We chose the NV center in diamond, which is widely considered a promising candidate for quantum sensing and has recently been applied to imaging high-pressure phase transitions~\cite{hsieh2019imaging} and superconducting systems~\cite{bhattacharyya2023imaging}. To obtain multireference states we employed a quantum defect embedding theory (QDET)~\cite{sheng2022green}, which allows us to derive an effective Hamiltonian $H_{\text{eff}}$, by defining an active space, see App.~\ref{app:qdet} for details. This effective Hamiltonian is used as the starting point of the QC-AFQMC calculation.

Although it has been shown~\cite{huang2023quantum} that a rather large active space is needed to fully converge the computed neutral excitation energies, here we only chose a minimum model of (4e, 3o) to carry out the QC-AFQMC calculations, as a proof of principle. The ground state has a $^3 A_2$ irrep due to the $C_{3v}$ symmetry of the defect and its wavefunction can be written as
\begin{equation}
    \big| ^3 A_2\big\rangle = \frac{1}{\sqrt{2}}\left(|a_1 \overline{a_1} e_x \overline{e_y}\rangle + |a_1 \overline{a_1} \overline{e_x} e_y\rangle\right),
\end{equation}
where $a_1, e$ denotes the irrep of the single-particle orbital and the bar symbol represents the spin-down channel. In Ref.~\cite{huang2022simulating}, a subset of the current authors investigated this system using VQE and found that error mitigation techniques were necessary to obtain an accurate ground state energy. Here, we use the tailored UCCSD ansatz in App.~\ref{app:quantum_simulations} as the quantum trial state for our QC-AFQMC calculations.

In QC-AFQMC, we used 4000 shadow circuits (with 100 shots each), 4800 walkers and an imaginary timestep of $\Delta \tau = 0.4\; \text{H.a.}^{-1}$. As can be seen from Fig.~\ref{fig:nv_center_afqmc}, the results obtained on quantum hardware (orange curve) agree with the noiseless reference (blue curve), which is within our expectation, given the noise resilience discussed in Sec.~\ref{subsection:noise_resilience}. Both curves converge to the classical reference FCI limit, which has been renormalized to zero.

\begin{figure}[hbt!]
    \centering
    \includegraphics[width=0.48\textwidth]{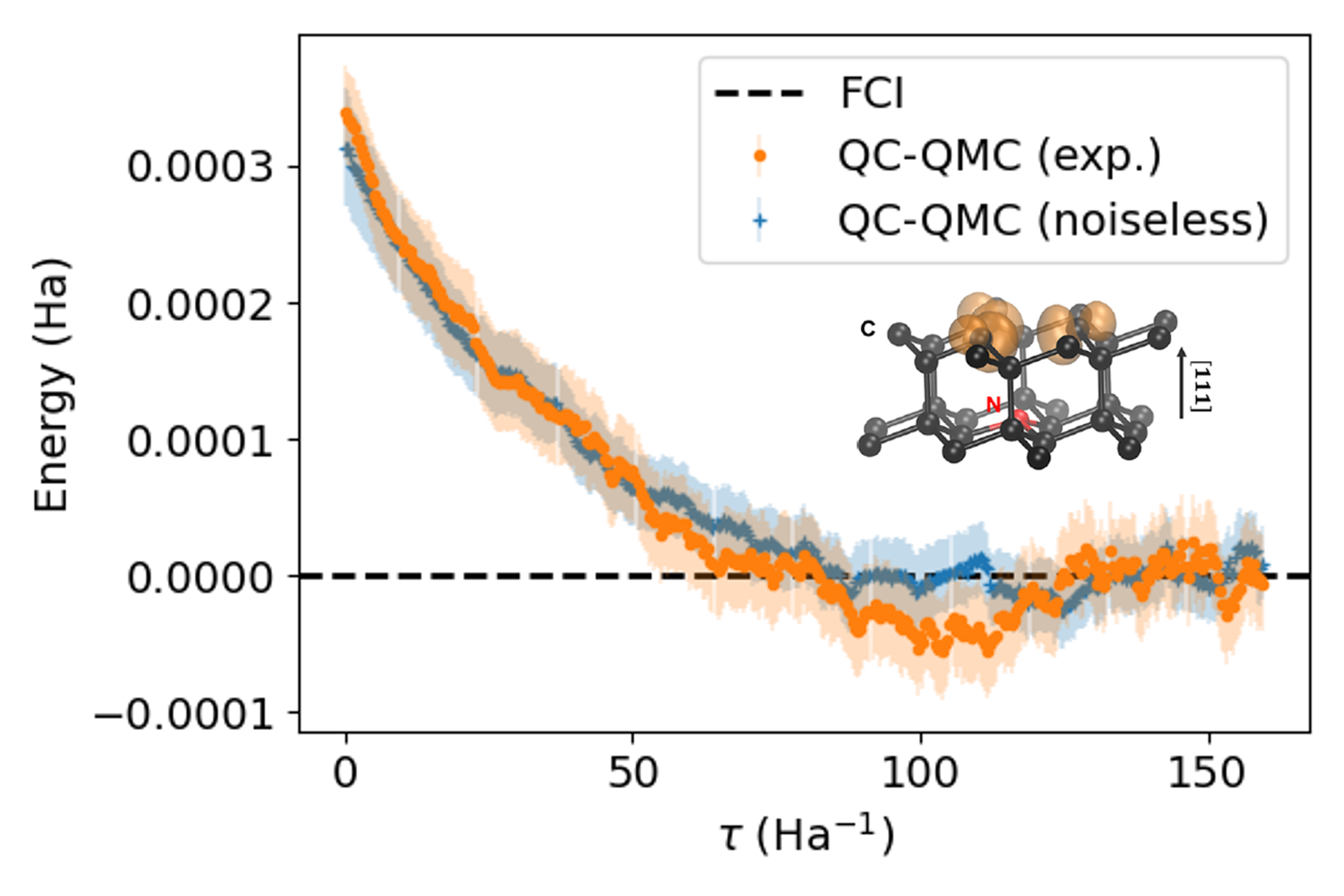}
    \caption{The QC-AFQMC calculation of an NV center in diamond using a noiseless quantum simulator (blue crosses), and the IonQ Aria quantum computer (orange circles). Converged results of both have an error on the order of 0.1 mHa compared to classical reference. The inset shows an atomistic model of the defect center.}
    \label{fig:nv_center_afqmc}
\end{figure}

The computation was again fully parallelized, with one walker per physical CPU core (4800 physical cores total). The total runtime required for 400 timesteps was approximately 1.5 hours.

\section{Discussion} \label{Discussions}

\begin{table*}[t]
\caption{\label{tab:Scaling_comparison} Comparison of different measurement schemes for QC-AFQMC.}
\renewcommand{\arraystretch}{1.8}
\setlength\tabcolsep{0pt}
\begin{tabular*}{\linewidth}{@{\extracolsep{\fill}}c|c|c|c}
\hline
\hline
Method for computing $\langle\Psi_T|\phi\rangle$ & Vacuum reference~\cite{xu2023quantum} & Clifford shadows~\cite{huggins2022unbiasing} & Matchgate shadows~\cite{wan2023matchgate} \\
\hline
Only offline QC access? & \tikzxmark & \tikzcmark & \tikzcmark \\
\hline
Circuit sampling complexity\footnotemark[1] & $\mathcal{O}\left(\frac{Mn^4}{\epsilon^2}\right)$ & $\mathcal{O}\left(\frac{\text{log}\left(Mn^4/\delta\right)}{\epsilon^2}\right)$ &$\mathcal{O}\left(\frac{\text{log}\left(Mn^4/\delta\right)}{\epsilon^2}\sqrt{n}\text{log}n\right)$ \\
\hline
Circuit depth complexity\footnotemark[2] & \multicolumn{3}{c}{$\mathcal{O}(N_T + n)$} \\
\hline
Noise-resilient overlap ratios & \tikzxmark & \tikzcmark & \tikzcmark \\
\hline
$\{\langle\Psi_T|\phi\rangle \}$ post-processing complexity & $\mathcal{O}(M)$ & $\mathcal{O}\left(\frac{\text{log}\left(Mn^4/\delta\right)}{\epsilon^2} M\text{exp}(n)\right)$ & $\mathcal{O}\left(\frac{\text{log}\left(Mn^4/\delta\right)}{\epsilon^2} Mn^{4.5} \text{log}n\right)$\\
\hline
$\{ \langle\Psi_T|H|\phi\rangle \}$ post-processing complexity\footnotemark[3] & $\mathcal{O}(Mn^4)$ & $\mathcal{O}\left(\frac{\text{log}\left(Mn^4/\delta\right)}{\epsilon^2} M\text{exp}(n)\right)$ & $\mathcal{O}\left(\frac{\text{log}\left(Mn^4/\delta\right)}{\epsilon^2} Mn^{8.5} \text{log}n\right)$\\
\hline
\hline
\end{tabular*}
\footnotetext[1]{$M$ in the scaling should be interpreted as the total number of walker states involved in the algorithm. The logarithmic factor entering the complexities for the shadow-based methods arises from the use of a median-of-means estimator, and is not present in the heuristic mean estimator used in this work.}
\footnotetext[2]{$\mathcal{O}(N_T)$ represents the quantum trial state preparation circuit depth complexity. And the $|\phi\rangle$ preparation, random Clifford, and Matchgate circuit for the three schemes all scale as $\mathcal{O}(n)$ in depth.}
\footnotetext[3]{The local energy numerator can be computed by decomposing it into a linear combination of $\mathcal{O}(n^4)$ overlap amplitudes, due to the complexity of the Hamiltonian. The number of shadow samples required also contributes to the post-processing cost.}
\end{table*}

As we discussed in Sec.~\ref{introduction}, variational algorithms like VQE have been the methods of choice for quantum chemistry problems on near-term quantum computers. The QC-AFQMC scheme discussed here is complementary to VQE. Specifically, QC-AFQMC could be viewed as an error mitigation technique for VQE, where both the inaccuracy of the ansatz used and the noise effects would in principle be corrected by imaginary time evolution. In turn, the quantum trial state used in the QC-AFQMC algorithm could come from VQE calculations (in a smaller active space, or for a reduced problem size). However, one notable difference from the VQE ansatz is that the quantum trial preparation circuit needs to satisfy $V_T|\mathbf{0}\rangle = |\mathbf{0}\rangle$, if following the approach used in this work. This additional constraint prevents the use of some popular ansatz circuits in VQE, such as qubit coupled cluster~\cite{ryabinkin2018qubit}, and circuit simplification techniques~\cite{hempel2018quantum}. For some ansatz circuits it may be possible to circumvent this limitation by introducing another ancilla qubit, and controlling the implementation of the ansatz (which can often be done by controlling select gates, rather than every gate), following the procedure in App.~E of Ref.~\cite{wu2023error}.

The possibility of achieving exponential quantum advantage in solving quantum chemistry problems remains controversial~\cite{lee2023evaluating}. This is arguably not the direct aim of QC-AFQMC, given that phaseless-AFQMC is already a polynomially scaling method. Instead, the motivation of ph-QC-AFQMC is that the trial state used may provide energies with smaller biases than the purely classical algorithm. A full comparison requires an improved understanding of the advantages offered by quantum trial states over classical trial states and experimenting with larger systems, which we leave for future investigations. As a proxy for this, we could compare the cost of classical methods using state-of-the-art techniques~\cite{pham2023scalable, chen2023hybrid, weber2023design, shee2023potentially}, such as Multi-Slater trial states. Two popular implementations using Wick's theorem have the following scaling for evaluating the local energy, $\mathcal{O}(MXn^3 + MXN_c)$~\cite{mahajan2021taming} and $\mathcal{O}(MXn^4 + MN_c)$~\cite{mahajan2020efficient}, where $M, X, N_c$ are the number of walkers, Cholesky vectors, and Slater determinants in the trial, respectively. This can be compared against the complexities of the quantum algorithm, which is summarized in Table.~\ref{tab:Scaling_comparison} for the three different proposed approaches for measuring the overlap amplitudes. We see that the classical algorithms are more efficient in $n$ than the classical shadows-based quantum algorithms, which scale as $\mathcal{O}(n^{8.5}\text{log}^2 n)$. Following this initial posting of this paper, Ref.~\cite{jiang2024unbiasing} proposed using algorithmic differentiation to compute the local quantities, which may improve the high polynomial scaling of classical post-processing. Integrating this approach with the quantum algorithm is an important direction for future research. Faced with the steep overhead of classical post-processing, the quantum trial state would need to lead to a much smaller bias than a classical multi-Slater trial state with a large number of determinants to be practically advantageous. It is currently an open question as to what kind of quantum trial state is best suited for QC-AFQMC. The only criterion adopted so far is the trial state energy $\langle\Psi_T|H|\Psi_T\rangle$, but there might exist other, better-suited criteria.

While Matchgate classical shadows are formally efficient, they introduce a large post-processing cost. This presents a practical challenge for scaling up the QC-AFQMC algorithm to larger system sizes. To highlight this, in Fig.~\ref{fig:resource_estimation} we present the classical runtime for post-processing a single timestep of the QC-AFQMC algorithm, assuming parallelization with 1 CPU core per walker. Numerical results were obtained for hydrogen and then extrapolated to larger system sizes using the scaling of the QC-AFQMC algorithm. The results show that our implementation of the algorithm is only realistic up to eight spin-orbitals (water). Within each CPU, computation associated with postprocessing the results of each of the shadow circuits can be further parallelized. While not investigated in this paper, parallelization between shadow circuits can potentially speed up the classical runtime. For example, using one million CPU cores~\cite{posey2018hpc} would drive the classical post-processing of benzene toward a realistic regime, of around 6 hours per timestep. Nevertheless, for larger systems such as FeMoco, further optimizations in the algorithm and/or classical post-processing are still required to become practical. Compared to the classical shadows approach, the Hadamard-test based approach of Ref.~\cite{xu2023quantum} reduces the cost of classical post-processing. However, this comes at a price of an increased number of quantum circuit repetitions (scaling linearly with the number of walkers), leading to similar measurement bottlenecks as observed for VQE methods~\cite{gonthier2022MeasurementsVQE}. In addition, this approach requires iterative communication between quantum and classical processors, which may be subject to latencies. As such, despite the higher classical post-processing costs of the classical shadow based approach, we currently view it as the more promising approach to QC-AFQMC.

\begin{figure}[hbt!]
    \centering
    \includegraphics[width=0.48\textwidth]{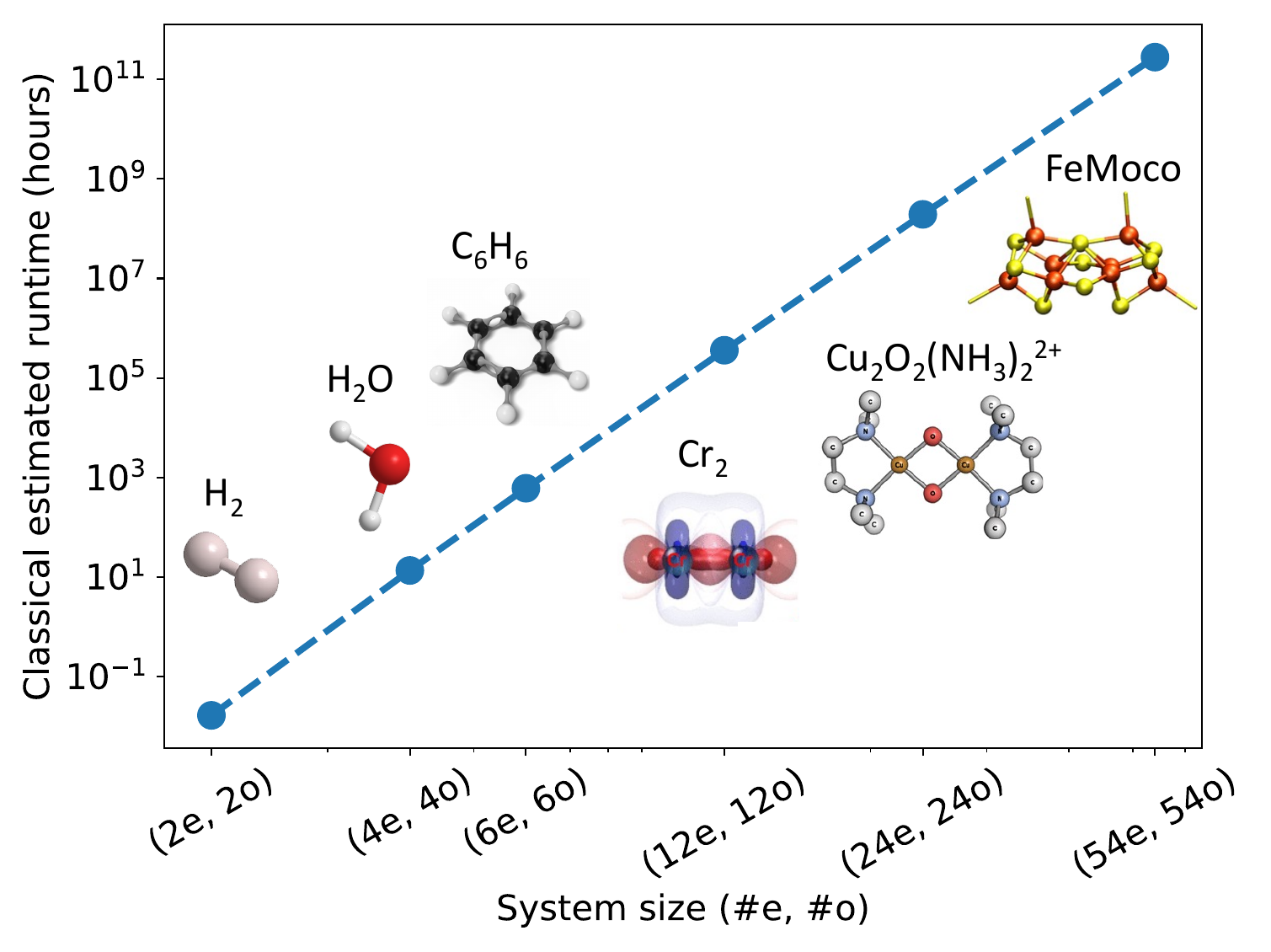}
    \caption{Prediction of the runtime estimation of classical post-processing for a single timestep in the QC-AFQMC algorithm performed on a single CPU core. These results were extrapolated from the runtime of post-processing results from hydrogen, based on a scaling of classical post-processing of $\mathcal{O}(n^{8.5}\log^2n)$. We make the optimistic assumptions that the error thresholds $\epsilon, \delta$ in classical shadows do not need to decrease as the system size increases. The scaling is dominated by the local energy estimation. The system sizes (plotted on the $x$ axis with a $\log$ scale) are quantified using active spaces, taken from the literature~\cite{mcardle2020quantum, ryabinkin2018qubit, goings2023molecular, anderson2020efficient, reiher2017elucidating}.}
    \label{fig:resource_estimation}
\end{figure}

Beyond consideration of computational efficiency, it is important to consider how noise affects the final energy estimate from QC-AFQMC. We have investigated above the effects of hardware noise on the evaluation of overlaps $\braket{\Psi_T}{\phi_i}$ and their ratios $\frac{\braket{\Psi_T}{\phi_i}}{\braket{\Psi_T}{\phi_j}}$, and we have found natural noise resilience for the latter. While this noise resilience will not persist in all settings (e.g. coherent state preparation noise, or noise that violates the GTM assumptions), our results show that there are no benefits to be gained from the robust classical shadows schemes. Since the AFQMC algorithm is driven by overlap ratios, it is important to understand how AFQMC performs with unbiased estimates of overlap ratios, subject to shot noise from finite measurement statistics. This question has recently been addressed by Ref.~\cite{kiser2023classical}, which observed that the error in the AFQMC energy estimate can be upper bounded by the 2-norm of the Hamiltonian matrix elements. Hence, we can conclude that if the noise obeys the conditions specified above, QC-AFQMC can still achieve accurate results, despite the presence of noise.

Finally, it has been acknowledged~\cite{huggins2022unbiasing} and highlighted~\cite{mazzola2022exponential} that the overlap amplitudes $\langle\Psi_T|\phi\rangle$ are expected to decay exponentially with system size, as discussed in Sec.~\ref{AFQMC}. In the worst case scenario, this decay implies an exponentially large number of measurements may still be required to control the uncertainty in estimating these overlaps. Establishing the practical viability of QC-AFQMC will require developing an improved understanding of these challenges, and their possible solutions.

\section{Conclusions} \label{Conclusions}

In this work, we carried out the first end-to-end experimental evaluation of the recently proposed Matchgate shadows~\cite{wan2023matchgate} powered QC-AFQMC algorithm for quantum chemistry~\cite{huggins2022unbiasing}. We observed that the algorithm is inherently noise robust, which we find to be a consequence of the natural noise resilience of evaluating overlap ratios via Matchgate shadows. We provided a theoretical explanation for this observed phenomenon, which also elucidates limitations in recently proposed robust Matchgate shadow protocols. We also developed improvements to those protocols which can mitigate state preparation noise, and may be of independent interest. Nevertheless, despite the tantalizing noise resilience of the algorithm, our optimized practical implementations have uncovered a number of challenges to the scalability of the algorithm. Most prominently, while the Matchgate shadows protocol is asymptotically efficient, its high degree polynomial scaling for evaluating the local energy necessitates significant parallel compute resources for classical post-processing. We estimate that for larger molecules, we would require significant amounts of classical post-processing, as shown in Fig.~\ref{fig:resource_estimation}. As such, future work should focus on developing new methods for efficiently computing the local energy in QC-AFQMC, with lower post-processing costs. Many open questions remain about the best trial states to use for QC-AFQMC, and how to ensure non-vanishing overlaps between the trial state and walker states. The merits of QC-AFQMC found in this work, together with these open research questions, motivate the importance of further study of this algorithm applied to larger system sizes, on real quantum devices.

\section{Acknowledgement}
We thank Senrui Chen, Steve Flammia, Liang Jiang, Andrew Zhao, Kianna Wan, William Huggins, Siyuan Chen, Yu Jin, Guan Wang, and Ji Liu for fruitful discussions. We thank Eric Kessler and Alexander Dalzell for comments on this manuscript. This work was in part supported by the Next Generation Quantum Science and Engineering (Q-NEXT) center that develops quantum science and engineering technologies. Q-NEXT is supported by the U.S. Department of Energy, Office of Science, National Quantum Information Science Research Centers. This research used resources of the Oak Ridge Leadership Computing Facility at the Oak Ridge National Laboratory, which is supported by the Office of Science of the U.S. Department of Energy under Contract No. DE-AC05-00OR22725. We thank the Amazon Braket team for facilitating the access to IonQ quantum computers through Amazon Web Service (AWS), and the access to HPC services on AWS. We acknowledge the use of IBM Quantum services for this work and to advanced services provided by the IBM Quantum Researchers Program. The views expressed are those of the authors and do not reflect the official policy or position of IBM or the IBM Quantum team.

\appendix

\section{Auxiliary-field quantum Monte Carlo} \label{app:AFQMC}

The auxiliary-field quantum Monte Carlo (AFQMC) algorithm is a stochastic implementation of the imaginary time evolution (ITE) process:
\begin{equation}
    e^{-\tau(H - E_0)}|\Psi_I\rangle = \left[e^{-\Delta\tau(H - E_0)}\right]^N |\Psi_I\rangle,\;\; \Delta\tau = \frac{\tau}{N}, \label{eq:ITE}
\end{equation}
where $E_0$ is some approximated ground state energy typically from mean-field calculations. ITE projects out the excited state components in the initial state $|\Psi_I\rangle$, and the true ground state $|\Psi_g\rangle$ is ultimately obtained as long as $|\Psi_I\rangle$ is not orthogonal to it. In general, it is often as difficult to realize the ITE in Eq.~\ref{eq:ITE} as solving the time-independent Schr\"odinger equation. To remedy this challenge, the Hubbard-Stratonovich transformation~\cite{hubbard1959calculation, stratonovich1957method} can be employed, where the two-body interaction in the following form is recast into one-body coupled to an auxiliary field:
\begin{equation}
    e^{\frac{\Delta \tau}{2} v^2_{\gamma}} = \int \frac{d x_{\gamma}}{\sqrt{2\pi}} e^{-\frac{x_\gamma^2}{2}} e^{\sqrt{\Delta \tau} x_{\gamma} v_{\gamma}},
\end{equation}
where $v_{\gamma}$ is a matrix that represents a one-body interaction, and $x_{\gamma}$ is the auxiliary field. To achieve that, the Hamiltonian, typically written in second quantized form, has to be rewritten into a Cholesky decomposed format
\begin{equation}
\begin{split}
    H & = H_0 + \sum_{i,j} h_{ij} a_i^{\dagger} a_j + \frac{1}{2}\sum_{i,j,k,l} V_{ijkl} a_i^{\dagger} a_j^{\dagger} a_l a_k \\
    & = H_0 + v_0 - \frac{1}{2} \sum_{\gamma} v_{\gamma}^2, \label{eq:afqmc_hamiltonian}
\end{split}
\end{equation}
where $H_0$ is a constant, $h, V$ are the one- and two-electron integrals, $v_0$ is the modified one-body interaction, $v_{\gamma} = i\mathcal{L}_{\gamma}$, and $\mathcal{L}_{\gamma}$ is the Cholesky vector of the two-body interactions satisfying
\begin{equation}
    \mathcal{L}_{\gamma} = \sum_{pq} L^{\gamma}_{pq} a^{\dagger}_p a_q,\;\;V_{ijkl} = \sum_{\gamma} L^{\gamma}_{ik} L^{\gamma *}_{lj}.
\end{equation}
Note that we've assumed a mean-field subtraction~\cite{motta2018ab} in the above Hamiltonian. With it, and the second-order Trotter formula~\cite{trotter1959product, suzuki1976relationship}: $e^{A+B} \approx e^{A/2} e^{B} e^{A/2}$, the short-time propagator can be approximated as
\begin{equation}
\begin{split}
    & e^{-\Delta\tau(H - E_0)}\\
    & \approx e^{-\Delta\tau(H_0 - E_0)} e^{-\frac{\Delta\tau}{2} v_0} \prod_{\gamma} e^{\frac{\Delta\tau}{2} v^2_{\gamma}} e^{-\frac{\Delta\tau}{2} v_0}\\
    & = \int d\textbf{x} p(\textbf{x}) B(\textbf{x}) + O(\Delta\tau^2), \label{eq:propagator}
\end{split}
\end{equation}
where we've defined
\begin{equation}
    B(\textbf{x}) = e^{-\frac{\Delta\tau}{2} v_0} \prod_{\gamma} e^{\sqrt{\Delta\tau} x_{\gamma} v_{\gamma}} e^{-\frac{\Delta\tau}{2} v_0}.
\end{equation}
The constant prefactor is usually left out and $p(\textbf{x})$ is the multi-variable standard normal distribution. After this transformation, we have only one-body terms left, and the original interacting problem is now mapped onto an ensemble of non-interacting problems coupled to a set of auxiliary fields. The probability distribution $p(\textbf{x})$ in the above expression is realized by an ensemble of walkers $|\phi_l\rangle$, which are evolved as
\begin{equation}
    |\phi^{(i+1)}\rangle = B(\textbf{x}) |\phi^{(i)}\rangle, \label{eq:walkers_evolution}
\end{equation}
where the superscript represents the timestep. The walkers in AFQMC are chosen as nonorthogonal Slater determinants and stay in the manifold of Slater determinants $\mathcal{S}$ during evolution, due to the Thouless theorem~\cite{thouless1960stability, thouless1961vibrational}. The ground state is then estimated by a mixed energy estimator which is exact and can lead to considerable simplifications in practice
\begin{equation}
    E = \frac{\langle \Psi_T|H e^{-N\Delta\tau(H - E_0)}|\phi^{(0)}\rangle}{\langle \Psi_T|e^{-N\Delta\tau(H - E_0)}|\phi^{(0)}\rangle}, \label{eq:mixed_estimator}
\end{equation}
where $|\Psi_T\rangle, |\phi^{(0)}\rangle$ are the trial state and initial Slater determinant, respectively.

So far we haven't imposed any constraints on the evolution of walkers, which is known as free-projection AFQMC. The only potential systematic error comes from the Trotterization error. However, free-projection AFQMC is prone to large fluctuations in the estimated energy because of an asymptotic instability in $\tau$, also known as the phase problem. The phase problem manifests as the phase $\theta_l$ of walkers evolves into the whole $[0, 2\pi)$ range during the random walk (e.g., see Fig. 3 in Ref.~\cite{zhang201315} for visualization), and results in the denominator in Eq.~\ref{eq:gs_estimator}, i.e., $\sum_l w_l e^{i\theta_l}$ approaching 0 exponentially quickly with $\tau$. This problem could be viewed as a generalization of the sign problem, see Refs.~\cite{motta2018ab, zhang201315} for a detailed discussion.

Importance sampling is a variance-reduction technique widely used in QMC methods. It's also employed to control the phase problem in AFQMC. To do that, we first note that the Eq.~\ref{eq:propagator} still holds up to a complex-valued shift $\overline{\textbf{x}}$. Therefore it can be modified as:
\begin{equation}
    p(\textbf{x}) \to p(\textbf{x} - \overline{\textbf{x}}),\;\;B(\textbf{x}) \to B(\textbf{x} - \overline{\textbf{x}}).
\end{equation}
To find the best $\overline{\textbf{x}}$, we first expand Eq.~\ref{eq:mixed_estimator} as
\begin{equation}
\begin{split}
    E & \approx \frac{\int\prod_{k=0}^{N-1}d\textbf{x}_k p(\textbf{x}_k - \overline{\textbf{x}}_k) \langle\Psi_T|H|\phi^{(N)}\rangle}{\int\prod_{k=0}^{N-1}d\textbf{x}_k p(\textbf{x}_k - \overline{\textbf{x}}_k) \langle\Psi_T|\phi^{(N)}\rangle}\\
    & = \frac{\int\prod_{k=0}^{N-1}d\textbf{x}_k p(\textbf{x}_k) I(\textbf{x}_k, \overline{\textbf{x}}_k, \phi^{(k)}) \frac{\langle\Psi_T|H|\phi^{(N)}\rangle}{\langle\Psi_T|\phi^{(N)}\rangle}}{\int\prod_{k=0}^{N-1}d\textbf{x}_k p(\textbf{x}_k) I(\textbf{x}_k, \overline{\textbf{x}}_k, \phi^{(k)})},
\end{split}
\end{equation}
where we've defined 
\begin{equation}
    |\phi^{(N)}\rangle = B(\textbf{x}_{N-1} - \overline{\textbf{x}}_{N-1})\dots B(\textbf{x}_0 - \overline{\textbf{x}}_0)|\phi^{(0)}\rangle,
\end{equation}
and an importance function
\begin{equation}
    I(\textbf{x}, \overline{\textbf{x}}, \phi) = \frac{\langle \Psi_T|B(\textbf{x} - \overline{\textbf{x}})|\phi\rangle}{\langle \Psi_T|\phi\rangle} e^{\textbf{x}\cdot\overline{\textbf{x}} - \frac{\overline{\textbf{x}}^2}{2}}.
\end{equation}
By choosing
\begin{equation}
    \overline{x}_{\gamma} = -\sqrt{\Delta\tau} \langle v_{\gamma}\rangle,\;\; \langle v_{\gamma}\rangle = \frac{\langle\Psi_T|v_{\gamma}|\phi\rangle}{\langle\Psi_T|\phi\rangle},
\end{equation}
fluctuations in the importance function to first order in $\sqrt{\Delta \tau}$ is cancelled~\cite{motta2018ab}, leading to a more stable random walk. The dynamic shift, usually referred to as a force bias, modifies the sampling of the auxiliary fields by shifting the center of the Gaussian distribution $p(\textbf{x})$ according to the overlap $\langle\Psi_T|\phi\rangle$. The importance function can therefore be further approximated as $I(\textbf{x}, \overline{\textbf{x}}, \phi) \approx \text{exp}\left[-\Delta\tau(E^{\text{loc}} - E_0)\right]$, known as the local energy formalism. Introducing the force bias into the importance sampling would solve the instability if $B(\textbf{x})$ were real, which could be true in special cases. For a general phase problem, however, it doesn't lead to complete control of the phase problem.

To further control the phase problem, the local energy in the importance function has to be replaced by its real part, leading to $\theta_l=0$ and $w_l > 0$ in Eq.~\ref{eq:gs_estimator}. A phaseless approximation is also adopted as
\begin{equation}
\begin{split}
    I(\textbf{x}, \overline{\textbf{x}}, \phi_l) & \approx \text{exp}\left[-\Delta\tau(\Re E^{\text{loc}}_l - E_0)\right]\\
    & \times \max\left(0, \cos\left[\arg\left(\frac{\langle \Psi_T|B(\textbf{x} - \overline{\textbf{x}})|\phi_l\rangle}{\langle \Psi_T|\phi_l\rangle}\right)\right]\right), \label{eq:phaseless}
\end{split}
\end{equation}
where abrupt phase changes during the random walk are now forbidden. The weight is updated as: $w_l^{(k)} \gets w_l^{(k-1)}\times I(\textbf{x}, \overline{\textbf{x}}, \phi_l^{(k-1)})$. Combining these techniques finally resolves the phase problem, and the whole algorithm is usually referred to as the phaseless AFQMC (ph-AFQMC). Note that the phaseless approximation also introduces a bias to the mixed energy estimator, which we seek to systematically improve by using quantum computation in QC-AFQMC.

The whole process of the AFQMC algorithm is summarized in Alg.~\ref{alg:afqmc}, following the presentation in Ref.~\cite{xu2023quantum}.
\begin{algorithm}[H]
\caption{Auxiliary-field quantum Monte Carlo}\label{alg:afqmc}
\begin{algorithmic}[1]
\State Input $H, E_0, |\Psi_T\rangle, |\Psi_I\rangle$, number of walkers $M$, number of timesteps $N$, and step size $\Delta\tau$. ($H$ will be decomposed into $L$ Cholesky vectors.)
\For{$l = 1$ to $M$}
\State $|\phi_l^{(0)}\rangle \gets |\Psi_I\rangle,\;\; w_l^{(0)} \gets 1$ \Comment{Initialize the walker.}
\For{$k = 1$ to $N$}
\For{$\gamma = 1$ to $L$}
\State $\overline{x}_{\gamma} = -\sqrt{\Delta\tau}\frac{\langle\Psi_T |\hat{v}_{\gamma}|\phi_l^{(k-1)}\rangle}{\langle\Psi_T|\phi_l^{(k-1)}\rangle}$
\EndFor
\State Sample $\textbf{x}$ according to distribution $p(\textbf{x})$
\State $|\phi_l^{(k)}\rangle \gets B(\textbf{x}-\overline{\textbf{x}})|\phi_l^{(k-1)}\rangle$. \Comment{Update the walker}
\State $E_l^{(k-1)}\gets \frac{\langle\Psi_T |H|\phi_l^{(k-1)}\rangle}{\langle\Psi_T|\phi_l^{(k-1)}\rangle}$ \Comment{Compute local energy}
\State $\theta \gets \arg\left(\frac{\langle\Psi_T |\phi_l^{(k)}\rangle}{\langle\Psi_T|\phi_l^{(k-1)}\rangle}\right)$ \Comment{Compute the phase}
\State $w_l^{(k)} \gets w_l^{(k-1)} \times I(\textbf{x}, \overline{\textbf{x}}, \phi_l^{(k-1)})$ \Comment{Update weight}
\EndFor
\State $E_l^{(N)} \gets \frac{\langle\Psi_T |H|\phi_l^{(N)}\rangle}{\langle\Psi_T|\phi_l^{(N)}\rangle}$\Comment{Compute the local energy}
\EndFor
\State Output the energy $E \gets \frac{\sum_{l}^{M} w_l^{(N)} E_l^{(N)}}{\sum_l^M w_l^{(N)}}$
\end{algorithmic}
\end{algorithm}

\section{Local energy evaluation with classical shadows}  \label{app:local_energy}

The most time-consuming step in QC-AFQMC is the evaluation of the local energy. Here we discuss how this task is performed in our implementation compatible with classical shadows. We first note that the local energy numerator can be rewritten as
\begin{equation}
\begin{split}
    & \langle \Psi_T|H|\phi\rangle = \langle \Psi_T|U_{\phi} U_{\phi}^{\dagger} H U_{\phi}|\Phi_0\rangle\\
    & = \sum_{pr} \langle \Psi_T|U_{\phi}|\Phi_p^r\rangle\langle \Phi_p^r| \Bar{H}|\Phi_0\rangle + \sum_{pqsr} \langle \Psi_T|U_{\phi}|\Phi_{pq}^{rs}\rangle\langle \Phi_{pq}^{rs}| \Bar{H}|\Phi_0\rangle\\
    & = \sum_{pr} \langle \Psi_T|\phi_p^r\rangle\langle \Phi_p^r| \Bar{H}|\Phi_0\rangle + \sum_{pqsr} \langle \Psi_T|\phi_{pq}^{rs}\rangle\langle \Phi_{pq}^{rs}| \Bar{H}|\Phi_0\rangle,
\end{split}
\end{equation}
where $\Bar{H} = U_{\phi}^{\dagger} H U_{\phi}$ represents the rotated Hamiltonian and $|\Phi_0\rangle,\; |\Phi_p^r\rangle (|\Phi_{pq}^{rs}\rangle)$ denotes the Hartree-Fock state and its single (double) excited counterpart. We also have $|\phi_p^r\rangle = U_{\phi} |\Phi_p^r\rangle$. The second line in the above equation utilizes the resolution of identity and the fact that Hamiltonian has only up to two-body interactions.

The Hamiltonian rotation can be done with $\mathcal{O}(n^5)$ operations. In the above equation, the overlap amplitudes $\langle \Psi_T|\phi_{p}^{r}\rangle, \langle \Psi_T|\phi_{pq}^{rs}\rangle$ can be calculated with Matchgate shadows each with complexity $\mathcal{O}(n^4)$, as already discussed in the main text. Each matrix element can be computed using the Slater-Condon rules with complexity $\mathcal{O}(1)$. Due to the number of possible excited Slater determinants, which scales as $\mathcal{O}(n^4)$, the scaling of the evaluation of the local energy for each walker would therefore scale as $\mathcal{O}(n^8 + n^5 + n^4) = \mathcal{O}(n^8)$ using this approach.

We also note that there exist other ways to evaluate the local energy with Matchgate shadows. Sec.V C of Ref.~\cite{wan2023matchgate} proposed an alternative way to evaluate the local energy using the Grassmann algebra. In this approach, the Hamiltonian is first rewritten into the Majorana formalism. Then the mixed estimator of each Majorana operator $\langle\Psi_T|\gamma_S|\phi\rangle$ can be estimated in a similar fashion to the overlap amplitude. This approach leads to a scaling between $\mathcal{O}(n^9)$, due to the relevant matrices being non-invertible, making it slightly less favorable than the previous approach.

\section{Classical shadows revisited} \label{app:shadow_tomography}

In this section, we'll revisit classical shadows from a group representation perspective, which will be useful when we include noise effects. We closely follow the presentation in Ref.~\cite{chen2021robust}.

We start by introducing the notations and conventions used throughout these appendices. We consider a $n$-qubit system with Hilbert space $\mathcal{H}_n$. Its dimension is denoted by $d \equiv 2^n$. All the $n$-qubit operators (super-operators) live in a vector space defined as $\mathcal{L}(\mathcal{H}_n)$ ($\mathcal{L}(\mathcal{L}(\mathcal{H}_n))$). For better clarity, we make use of the Liouville representation of operators and super-operators: 
\begin{enumerate}
    \item Operators are notated using the double kets $|A\rrangle = \frac{1}{\sqrt{\text{tr}(A^{\dagger}A)}}A$ for its length $d^2$ vectorization for $A \in \mathcal{L}(\mathcal{H}_n)$ and $\llangle B|A \rrangle := \frac{\text{tr}(B^{\dagger} A)}{\sqrt{\text{tr}(B^{\dagger}B)\text{tr}(A^{\dagger}A)}}$;
    \item Super-operators, written as cursive letters, are mapped to $d^2 \times d^2$ matrices: any $\mathcal{E} \in \mathcal{L}(\mathcal{L}(\mathcal{H}_n))$ can be specified by its matrix elements $\mathcal{E}_{ij} := \llangle B_i|\mathcal{E}|B_j\rrangle$, where $\{|B_i\rrangle\}$ is an orthonormal basis for $\mathcal{L}(\mathcal{L}(\mathcal{H}_n))$.
\end{enumerate}

For any unitary $U$, its corresponding channel is denoted by $\mathcal{U}(\cdot):= U(\cdot)U^{\dagger}$. For any $|\phi\rangle \in \mathcal{H}_n$, $|\phi\rrangle$ is the vectorization of $|\phi\rangle\langle\phi|$. And hats indicate statistical estimators, e.g., $\hat{o}$ denotes an estimate for $o = \text{tr}(O\rho)$.

Classical shadows are based on a simple measurement primitive: for the quantum state $\rho$, apply a unitary $U$ randomly drawn from a distribution of unitaries $\mathcal{D}$ and measure in the computational basis. This produces measurement outcomes $b \in \{0, 1\}^{\otimes n}$ with probability $\langle b|U \rho U^{\dagger}|b\rangle$. One then inverts the unitary on the outcome $|b\rangle$ in post-processing, which amounts to storing a classical representation of $U^{\dagger} |b\rangle$. The distribution $\mathcal{D}$, from which the random unitaries $U$ are drawn, usually can be associated with a group $G$, e.g., random Clifford unitaries and the Clifford group. $\mathcal{M}$ can be viewed as a twirl of the measurement channel $\mathcal{M}_Z$, giving:
\begin{equation}
    \mathcal{M} := \underset{g\sim G}{\mathbb{E}} \mathcal{U}_g^{\dagger} \mathcal{M}_Z \mathcal{U}_g,\;\; \mathcal{M}_Z = \sum_{b\in\{0, 1\}^{n}} |b\rrangle\llangle b|,
\end{equation}
where $g$ is an element in $G$, and the random unitaries $\mathcal{U}$ is a unitary representation of $G$. Applying Schur's lemma, and assuming no multiplicities in $\mathcal{U}$ lead to
\begin{equation}
    \mathcal{M} = \sum_{\lambda \in R_G} f_{\lambda}\Pi_{\lambda}, \;\;\;\; f_{\lambda} = \frac{\text{tr}(\mathcal{M}_Z \Pi_{\lambda})}{\text{tr}(\Pi_{\lambda})}, \label{eq:quantum_channel}
\end{equation}
where $R_G$ are the irreducible representations (irreps) of $G$. The super-operators $\Pi_{\lambda} \in \mathcal{L}(\mathcal{L}(\mathcal{H}_n))$ are orthogonal projectors onto the irreducible subspaces $\Gamma_{\lambda}$, and $f_{\lambda}$ is the eigenvalue associated with each orthogonal projector.

\section{Matchgate shadows} \label{app:Matchgates}

In this section, we outline the basics of Matchgate shadows, closely following the presentation in Ref.~\cite{wan2023matchgate}. The Matchgate unitaries are based on the Majorana formalism. For a system of $n$ orbitals, $2n$ Majorana operators can be defined as
\begin{equation}
    \gamma_{2j-1} = a_j + a_j^{\dagger},\;\;\;\; \gamma_{2j} = -i(a_j - a_j^{\dagger}),
\end{equation}
for $j \in [n]:= \{1, \dots, n\}$. These have qubit representation
\begin{equation}
    \gamma_{2j-1} = \left(\Pi_{i=1}^{j-1} Z_i\right) X_j,\;\; \gamma_{2j} = \left(\Pi_{i=1}^{j-1} Z_i\right) Y_j, \label{eq:majorana_to_pauli}
\end{equation}
under the Jordan-Wigner (JW) transformation~\cite{wigner1928paulische}, where $X_i, Y_i, Z_i$ are the Pauli operators on qubit $i$. For a subset of indices $S \subseteq [2n]$, we denote by $\gamma_S$ the product of the Majorana operators indexed by the elements in $S$ in ascending order. That is,
\begin{equation}
    \gamma_S := \gamma_{\mu_1} \dots \gamma_{\mu_{|S|}},    
\end{equation}
for $S = \{\mu_1,\dots,\mu_{|S|}\} \subseteq [2n]$ with $\mu_1 < \dots < \mu_k$.

\subsection{Fermionic Gaussian states}
Matchgate circuits are qubit representations of fermionic Gaussian unitaries under the JW transformation. The fermionic Gaussian unitaries transform between valid sets of Majorana operators $\{\gamma_{2j-1}, \gamma_{2j}\}, j\in [n]$. Fermionic Gaussian states are the ground states and thermal states of non-interacting fermionic Hamiltonians. An $n$-mode Gaussian state is any state whose density operator $\varrho$ can be written as
\begin{equation}
    \varrho = \prod_{j=1}^{n}\frac{1}{2}\left(I - i\lambda_j \gamma_{2j-1}\gamma_{2j}\right),
\end{equation}
for some coefficients $\lambda_j \in [-1, 1]$. If $\lambda_j \in \{-1, 1\}$ for all $j \in [n]$, then $\varrho$ is a pure Gaussian state; otherwise, $\varrho$ is a mixed state. For any computational basis state $|b\rangle$, for example, we have
\begin{equation}
    |b\rangle\langle b| = \prod_{j=1}^n \frac{1}{2}\left[I - i(-1)^{b_j} \gamma_{2j-1}\gamma_{2j}\right].
\end{equation}
A Gaussian state can also be defined by its two-point correlations $\text{tr}(\varrho \gamma_{\mu} \gamma_{\nu})$, which is also known as its covariance matrix. An $n$-qubit Gaussian state $\rho$ has its covariance matrix defined as:
\begin{equation}
    (C_{\rho})_{\mu\nu} := -\frac{i}{2}\text{tr}([\gamma_{\mu}, \gamma_{\nu}]\rho)
\end{equation}
for $\mu, \nu \in [2n]$, which is a $2n \times 2n$ antisymmetric matrix. For any computational basis state $|b\rangle\langle b|$ as an example, its covariance matrix is:
\begin{equation}
    C_{|b\rangle} := \bigoplus_{j=1}^{n} \begin{pmatrix}
        0 & (-1)^{b_j}\\
        (-1)^{b_j +1} & 0
    \end{pmatrix}.
\end{equation}

A Gaussian state which is also an eigenstate of the number operator $\sum_{j=1}^n a_j^{\dagger} a_j$ is known as Slater determinant. Any $\zeta$-fermion Slater determinant $|\varphi\rangle$ can also be written as
\begin{equation}
    |\phi\rangle = a'^{\dagger}_1\dots a'^{\dagger}_{\zeta} |0\rangle, \;\; a'^{\dagger}_j = \sum_{k=1}^n V_{kj} a^{\dagger}_k = U_V a^{\dagger}_j U_V^{\dagger},
\end{equation}
for some $n\times n$ unitary matrix $V$. Hence a Slater determinant can also be specified by the first $\zeta$ columns of $V$. Finally the Majorana operators $\{\gamma'\}_{\mu\in [2n]}$ have the following relation
\begin{align}
    \gamma'_{2j-1} & = \sum_k\left[\Re(V_{jk}) \gamma_{2j-1} - \Im(V_{jk}) \gamma_{2j}\right],\\
    \gamma'_{2j} & = \sum_k\left[\Im(V_{jk}) \gamma_{2j-1} + \Re (V_{jk})\gamma_{2j}\right].
\end{align}
So the fermionic Gaussian unitary $U_{Q'}$ that implements this transformation is given by the orthogonal matrix
\begin{equation}
    Q' = \begin{pmatrix}
        R_{11} & \dots & R_{1n}\\
        \vdots & \ddots & \vdots\\
        R_{n1} & \dots & R_{nn}
    \end{pmatrix},\; R_{jk} := \begin{pmatrix}
        \Re(V_{jk}) & -\Im(V_{jk})\\
        \Im(V_{jk}) & \Re(V_{jk})
    \end{pmatrix} \label{eq:Q_prime}.
\end{equation}

\subsection{Matchgate 3-design}
The adjoint action of Fermionic Gaussian unitaries $U_Q$ on the Majorana operators obeys
\begin{equation}
    U_Q \gamma_{\mu} U_Q^{\dagger} = \sum_{\nu \in [2n]} Q_{\nu \mu} \gamma_{\nu},\;\; \mu \in [2n], \label{eq:orthogonal_transformation}
\end{equation}
where $Q$ belongs to the orthogonal group $O(2n)$. Matchgate circuits form a continuous group $\text{M}_n$ which is in one-to-one correspondence with the orthogonal group $O(2n)$, up to a global phase:
\begin{equation}
    \text{M}_n = \{U_Q: Q \in O(2n)\}.
\end{equation}
Ref.~\cite{wan2023matchgate} considered two distributions: i). “uniform” distribution over $\text{M}_n$, where uniformity is more precisely given by the normalized Haar measure $\mu$ on $O(2n)$; ii). uniform distribution over the discrete Borel group $B(2n)$, which is the intersection of $\text{M}_n$ and the $n$-qubit Clifford group $\text{Cl}_n$. The second consists of $2n \times 2n$ signed permutation matrices:
\begin{equation}
    \text{M}_n \cap \text{Cl}_n = \{U_Q: Q \in B(2n)\}.
\end{equation}

For $j \in \mathbb{Z}_{>0}$, we use $\mathcal{E}^{(j)}_{\text{M}_n}$ and $\mathcal{E}^{(j)}_{\text{M}_n\cap \text{Cl}_n}$ to denote the $j$-fold twirl channels corresponding to the distributions over $\text{M}_n$ and $\text{M}_n\cap \text{Cl}_n$, respectively:
\begin{align}
    \mathcal{E}^{(j)}_{\text{M}_n} & := \int_{O(2n)} d\mu(Q) \mathcal{U}_Q^{\otimes j},\\
    \mathcal{E}^{(j)}_{\text{M}_n\cap \text{Cl}_n} & := \frac{1}{|B(2n)|}\sum_{Q \in B(2n)} \mathcal{U}_Q^{\otimes j},
\end{align}
where we've used the fact that 2-fold Matchgate twirl is hermitian~\cite{wan2023matchgate}. Since the measurement channel $\mathcal{M}$ in the classical shadows procedure and the variance of the estimates obtained from the classical shadows are determined by the 2- and 3-fold twirls, it is necessary to compare $\mathcal{E}^{(j)}_{\text{M}_n}$ and $\mathcal{E}^{(j)}_{\text{M}_n\cap \text{Cl}_n}$ up to $j=3$.

Ref.~\cite{wan2023matchgate} proved that the $j$-fold twirl channels of these two distributions are equivalent up to $j=3$. Thus, the discrete ensemble of Clifford Matchgate circuits is a 3-design for the continuous Haar-uniform distribution over all Matchgate circuits, in the same way that the Clifford group is a unitary 3-design~\cite{zhu2017multiqubit}. Ref.~\cite{wan2023matchgate} stated this result informally as: “The group of Clifford Matchgate circuits forms a ‘Matchgate 3-design’~”. In the implemented Matchgate shadow protocol, we adopt the latter since it requires less memory to store a random signed permutation matrix compared to an orthogonal matrix. An additional virtue of sampling in $\text{M}_n \cap \text{Cl}_n$ is that Clifford circuits can be randomly compiled to convert coherent errors into incoherent errors~\cite{hashim2021randomized} (within the shadow circuit), leading to better satisfaction of the GTM assumption of noise. Our implementation of Matchgate circuits is given in App.~\ref{app:implementation}.

\subsection{Matchgate channel}
The $n$-qubit Matchgate group $\text{M}_n$ has $(2n+1)$ irreps, so we can define correspondingly $(2n+1)$ subspaces
\begin{equation}
    \Gamma_k := \text{span}\left\{\gamma_S: S\in\begin{pmatrix} [2n]\\k \end{pmatrix}\right\},\;\; k\in \{0,\dots,2n\},
\end{equation}
where $\left(\begin{smallmatrix}[2n]\\k\end{smallmatrix}\right)$ denotes the set of subsets of $[2n]$ of cardinality $k$. Therefore, we have $\mathcal{L}(\mathcal{H}_n) = \oplus_{k=0}^{2n}\Gamma_k$. For fermionic simulations, we usually deal with even operators so we only look at the even subspaces $\Gamma_{\text{even}} = \oplus_{l=0}^n \Gamma_{2l}$.

Now we can derive the details of the Matchgate channel, which will serve as the basis for calculations involving noise presented in the following section. As we discussed above, the measurement channel for any shadow scheme can be derived by specifying the 2-fold twirl $\mathcal{E}^{(2)}$
\begin{equation}
    \mathcal{M}(\rho) = \text{tr}_1 \left[\sum_{b\in \{0,1\}^n} \mathcal{E}^{(2)}_{\text{M}_n\cap \text{Cl}_n} \left(|b\rangle\langle b|^{\otimes 2}\right)(\rho \otimes I)\right],
\end{equation}
where $\text{tr}_1$ denotes the partial trace over the first tensor component. The 2-fold twirl for the Matchgate group has the following form
\begin{gather}
    \mathcal{E}^{(2)}_{\text{M}_n\cap \text{Cl}_n} = \sum_{k=0}^{2n} |\Upsilon_k^{(2)}\rrangle \llangle \Upsilon_k^{(2)}|,\\ 
    |\Upsilon_k^{(2)}\rrangle = \begin{pmatrix}
        2n\\ k
    \end{pmatrix}^{-1/2} \sum_{S \in \left(\begin{smallmatrix}
        [2n]\\ k
    \end{smallmatrix}\right)} |\gamma_S\rrangle |\gamma_S\rrangle.
\end{gather}
The standard derivation can be found in Section IV.B in Ref.~\cite{wan2023matchgate}. For our purpose, we'll take a shortcut by assuming all the projectors $\Pi_{2l} \in \mathcal{L}(\mathcal{L}(\mathcal{H}_n))$ are already known
\begin{equation}
    \Pi_{2l} := \sum_{S\in\left(\begin{smallmatrix} [2n] \\2l\end{smallmatrix}\right)} |\gamma_S\rrangle\llangle \gamma_S|,
\end{equation}
where we've defined $|\gamma_S\rrangle = \gamma_S/\sqrt{2^n}$, and $\llangle\gamma_S|\gamma_{S'}\rrangle = \delta_{SS'}$. We directly start from Eq.~\ref{eq:quantum_channel}, and the eigenvalues are hence computed as
\begin{equation}
\begin{split}
    f_{2l} & = \frac{\text{Tr}\left[\mathcal{M}_Z \Pi_{2l} \right]}{\text{Tr}\left[\Pi_{2l}\right]}\\
    & = \frac{\sum_{S\in\left(\begin{smallmatrix}
        [2n]\\ 2l
    \end{smallmatrix}\right)}\sum_{b\in\{0,1\}^n} \left|\llangle \gamma_S|b\rrangle\right|^2}{\sum_{S\in \left(\begin{smallmatrix}
        [2n]\\ 2l
    \end{smallmatrix}\right)}\llangle \gamma_S|\gamma_S \rrangle}\\
    & = \begin{pmatrix}
        2n\\2l
    \end{pmatrix}^{-1} \sum_{b\in\{0,1\}^n} \sum_{T\in \left(\begin{smallmatrix}
        [n]\\l
    \end{smallmatrix}\right)} \left|\frac{(-i)^{|T|}}{\sqrt{2^n}}(-1)^{\sum_{j\in T} b_j}\right|^2\\
    & = \begin{pmatrix}
        2n\\ 2l
    \end{pmatrix}^{-1} \begin{pmatrix}
        n\\ l
    \end{pmatrix},
\end{split}
\end{equation}
where we've used
\begin{equation}
    |b\rrangle = \frac{1}{\sqrt{2^n}} \sum_{T\in \left(\begin{smallmatrix}
        [n]\\l
    \end{smallmatrix}\right)} (-i)^{|T|}(-1)^{\sum_{j\in T} b_j} |\gamma_{\text{pairs}(T)}\rrangle,
\end{equation}
and pairs($T$) is defined as $\bigcup_{j\in T}\{2j-1, 2j\}$. We only consider the even subspaces for fermionic simulations. Therefore the Matchgate channel is
\begin{equation}
    \mathcal{M} = \sum_{l=0}^n \begin{pmatrix}
        2n\\ 2l
    \end{pmatrix}^{-1} \begin{pmatrix}
        n\\ l
    \end{pmatrix} \Pi_{2l}. \label{eq:Matchgate_channel}
\end{equation}

With the above definitions and properties of fermionic Gaussian transformation, we can finally discuss how to use the Matchgate shadow to compute the desired overlap integrals. The overlap can first be reconstructed as
\begin{equation}
\begin{split}
    \langle\Psi_T|\phi\rangle & = 2\sum_{l=0}^n \begin{pmatrix}
        2n\\ 2l
    \end{pmatrix} \begin{pmatrix}
        n\\ l
    \end{pmatrix}^{-1}\\
    & \times \underset{Q \sim B(2n)}{\mathbb{E}} \text{tr} \left[|\phi\rangle\langle 0| \Pi_{2l} \left(U_Q^{\dagger} |\hat{b}\rangle\langle \hat{b}| U_Q\right)\right],
\end{split}
\end{equation}
where the trace involving $\Pi_{2l}$ correspond to the coefficient of $z^l$ in the following polynomial~\cite{wan2023matchgate}
\begin{equation*}
\begin{split}
    q_{|\phi\rangle, |b\rangle}(z) & = \frac{i^{\zeta/2}}{2^{n-\zeta/2}}\\
    & \times \text{pf}\Bigl[\bigl(C_{|0\rangle} + zW^* Q'^T Q^T C_{|b\rangle} Q Q' W^{\dagger}\bigr)\big|_{\overline{S}_{\zeta}}\Bigr].
\end{split}
\end{equation*}
The pf in the above equation denotes the matrix Pfaffian, and $M|_S$ represents matrix $M$ restricted to rows and columns in set $S$. $\zeta$ is the particle number, $C_{|0\rangle}$ denotes the covariance matrix of the vacuum state, and $Q'$ is an orthogonal matrix defined from the Slater determinant $|\varphi\rangle$ according to Eq.~\ref{eq:Q_prime}. We've also defined the $W$ matrix
\begin{equation}
    W = \bigoplus_{j=1}^{\zeta}\frac{1}{\sqrt{2}}\begin{pmatrix}
        1 & -i\\
        1 & i
    \end{pmatrix}\bigoplus_{j=\zeta+1}^{n}\begin{pmatrix}
        1 & 0\\
        0 & 1
    \end{pmatrix},
\end{equation}
and $\overline{S}_{\zeta} := [2n] \backslash \{1, 3, \dots, 2\zeta-1\}$. All these coefficients can be computed using polynomial interpolation in $\mathcal{O}((n-\zeta/2)^4)$ time, where $\zeta$ is the number of electrons.

\subsection{Noise resilience} \label{app:eigen-operator}

In this subsection, we prove Theorem~\ref{Theorem:MainTheorem}. This result serves as the foundation of noise resilience observed in experiments performed on quantum hardware.
\MainTheorem*
\begin{proof}
We expand
\begin{equation}
    \mathcal{P}_{0, \zeta}\big[\Pi_{2l} \left(|\phi\rangle\langle \mathbf{0}|\right)\big] = \sum_i c_i \mathcal{P}_{0, \zeta}\big[\Pi_{2l} \left(|i\rangle\langle \mathbf{0}|\right)\big].
\end{equation}
By the results of Lemma~\ref{Lemma:MainLemma},
\begin{equation}
\begin{split}
    \mathcal{P}_{0, \zeta}\big[\Pi_{2l} \left(|\phi\rangle\langle \mathbf{0}|\right)\big] & = \sum_i c_i b_{2l} \ket{i}\bra{\mathbf{0}} \\
    & = b_{2l} \ket{\phi}\bra{\mathbf{0}}
\end{split}    
\end{equation}
as required.
\end{proof}

The above result is proved using the following Lemma. The general proof idea is to write
\begin{equation}
    \Pi_{\zeta} \left( \ket{i}\bra{\mathbf{0}} \right) = 2^{-n} \sum_{S\in\left(\begin{smallmatrix} [2n] \\ 2l \end{smallmatrix}\right)} \bra{\mathbf{0}} \gamma_S^\dag \ket{i} \gamma_S.
\end{equation}
The action of the Majorana strings is to `create electrons' in the $\ket{\mathbf{0}}$ state. Terms in the sum evaluate to zero unless $\zeta$ electrons are created. For small/large numbers of Majorana operators in each string, insufficient electrons are created, and the whole sum evaluates to zero. In the intermediate regime, each Majorana string can be divided into two parts; a part that creates $\zeta$ electrons in the correct positions, and a part that contributes an ultimately unimportant phase and increases the multiplicity of a given term. The more interesting part is the part that creates the electrons, which acts as a projector from an $x$-electron subspace to an $x+\zeta$ electron subspace. The action of projecting into the subspace spanned by the vacuum state and computational basis states with Hamming weight $\zeta$ retains only terms that project from zero electrons to $\zeta$-electrons. Combinatorics give the desired values of $b_{2l}$.

\begin{lemma}\label{Lemma:MainLemma}
Given an $n$ qubit computational basis state $\ket{i}$ with even Hamming weight $\zeta$, then $\mathcal{P}_{0, \zeta} \big[\Pi_{2l} \left( \ket{i}\bra{\mathbf{0}} \right) = b_{2l} \ket{i}\bra{\mathbf{0}}$ with
\begin{equation}
    b_{2l}=\begin{cases}
	2^{\zeta - n} \binom{n-\zeta}{l-\frac{\zeta}{2}}, & \text{if $\frac{\zeta}{2} \leq l \leq n - \frac{\zeta}{2}$}\\
        0, & \text{otherwise} \end{cases}
\end{equation}
where $\mathcal{P}_{0, \zeta}$ is the projection from onto the subspace spanned by the vacuum state and computational basis states with Hamming weight $\zeta$.
\end{lemma}
\begin{proof}
First expand
\begin{equation}
    \Pi_{\zeta} \left( \ket{i}\bra{\mathbf{0}} \right) = 2^{-n} \sum_{S\in\left(\begin{smallmatrix} [2n] \\ 2l \end{smallmatrix}\right)} \bra{\mathbf{0}} \gamma_S^\dag \ket{i} \gamma_S
\end{equation}
using the definition of super-operators from App.~\ref{app:shadow_tomography}. The set $S$ contains the $\binom{2n}{2l}$ possible Majorana strings formed by choosing $2l$ Majorana operators acting on $n$ qubits. Recall the qubit representation of Majorana operators under the Jordan-Wigner transform
\begin{equation*}
    \gamma_{2j-1} = \left(\Pi_{i=1}^{j-1} Z_i\right) X_j,\;\; \gamma_{2j} = \left(\Pi_{i=1}^{j-1} Z_i\right) Y_j,
\end{equation*}
We refer to `connected' Majorana operators as pairs of the form $\gamma_{k} \gamma_{k+1} = iZ_{\frac{k+1}{2}}$ for $k$ odd. We refer to `unconnected' Majorana operators as pairs of the form $\gamma_{i} \gamma_{j}$ such that $j \neq i+1$ if $i$ is odd. Unconnected Majorana operators take the form $\gamma_{i} \gamma_{j} = p_{i,j} O_{i} \vec{Z}_{i+1}^{j-1} O_{j}$ where $p_{i,j} \in \{\pm 1, \pm i\}$ and $O \in \{X, Y\}$, and the notation $\vec{Z}_a^b := \bigotimes_{k=a}^b Z_k$. Hence, observe that connected Majorana operators act as the identity on $\ket{\mathbf{0}}$, while unconnected Majorana operators act as $\gamma_{i} \gamma_{j} \ket{\mathbf{0}} = \ket{0...0 1_i 0...0 1_j 0...0}$ (up to a complex phase). Because each unconnected pair increases the Hamming weight by 2, we require $\gamma_S$ to contain exactly $\zeta/2$ unconnected pairs, and all other pairs to be connected pairs. If this requirement is not satisfied, $\bra{\mathbf{0}} \gamma_S^\dag \ket{i} = 0$. \\

First consider the case $2l < \zeta$. In this case, $\gamma_S$ contains $2l < \zeta$ Majorana operators, and it is impossible to form $\zeta/2$ unconnected pairs. Hence, every term in the sum evaluates to zero, and $b_{2l}=0$, as required.\\

For $2l > 2n - \zeta$, observe that because the Hamming weight of $\ket{i}$ is $\zeta$, there are $(n-\zeta)$ $0$'s in $\ket{i}$. After placing a connected Majorana pair on each of these $0$ qubits, the remaining number of Majorana operators to distribute onto the $\zeta$ 1's in $\ket{i}$ is given by $2l - 2(n-\zeta) > 2n - \zeta - 2(n-\zeta) = \zeta$. It is necessary to place greater than $\zeta$ Majorana operators on the $\zeta$ qubits. As a result, at least two of the Majorana's will form a connected pair. This leaves fewer than $\zeta/2$ unconnected pairs, so $\bra{\mathbf{0}} \gamma_S^\dag \ket{i} = 0$. Hence, every term in the sum evaluates to zero, and $b_{2l}=0$, as required.\\

Proving the result for $\frac{\zeta}{2} \leq l \leq n - \frac{\zeta}{2}$ is more involved, and proceeds via explicit calculation. As discussed above, terms in the sum are only non-zero if $\gamma_S$ contains exactly $\zeta/2$ unconnected pairs. Denote the $\zeta$ qubits with value 1 in $\ket{i}$ as $[i] = [i_1, ..., i_\zeta]$, with $i_\alpha < i_\beta$ if $\alpha<\beta$. Denote the remaining $(n-\zeta)$ qubits with value 0 in $\ket{i}$ as $[j] = [j_1, ..., j_{n-\zeta}]$. Observe that $[i] \cap [j] = \emptyset$. 

The only non-zero terms in the sum correspond to strings with $\zeta$ Majorana operators forming $\zeta/2$ unconnected pairs on $[i]$ and $(2l -\zeta)$ Majorana operators forming $l-\frac{\zeta}{2}$ connected pairs on $[j]$. First consider the part of the string acting on $[i]$. The string $\gamma_{i_1} \gamma_{i_2} ... \gamma_{i_\zeta}$ can be written as a tensor product of non-overlapping (commuting) Pauli strings:
\begin{equation}
    \bigotimes_{\substack{k=1\\k~\text{odd}}}^{\zeta-1} p_{i_k,i_{k+1}} O_{i_k} \vec{Z}_{(i_k+1)}^{(i_{k+1}-1)} O_{i_{k+1}}.
\end{equation}
Note that the qubits acted on by $Z$ are not in the set $[i]$. Next, consider the part of the string acting on $[j]$. There are $\binom{n-\zeta}{l-\frac{\zeta}{2}}$ (paired) locations to place the connected pairs, denote any such choice as $[\bar{j}]$. The `connected part' of the string can be written as $\bigotimes_{\substack{k=1}}^{l-\frac{\zeta}{2}} \imath Z_{\bar{j}_k}$, where here we have used $\imath = \sqrt{-1}$ to avoid confusion with the index $i$.

Then explicitly compute the sum over these Majorana strings (qubits that are not explicitly acted on in the following expressions are implicitly acted on with the identity operator): 
\begin{widetext}
\begin{equation}
\begin{split}
    \Pi_{2l} \left( \ket{i}\bra{\mathbf{0}} \right) & = 2^{-n} \sum_{S\in\left(\begin{smallmatrix} [2n] \\ 2l \end{smallmatrix}\right)} \bra{\mathbf{0}} \gamma_S^\dag \ket{i} \gamma_S \\
    & = 2^{-n} \sum_{[\bar{j}] \in \binom{n-\zeta}{l-\frac{\zeta}{2}}} \sum_{\substack{O_{i_1} \in \{X_{i_1} Y_{i_1}\} \\ O_{i_2} \in \{X_{i_2},Y_{i_2}\} \\ ... \\ O_{i_\zeta} \in \{X_{i_\zeta},Y_{i_\zeta}\}}} \bra{\mathbf{0}} \left( \bigotimes_{\substack{k=1 \\ k~\text{odd}}}^{\zeta-1} p_{i_k,i_{k+1}} O_{i_k} \vec{Z}_{(i_k+1)}^{(i_{k+1}-1)} O_{i_{k+1}}  \bigotimes_{\substack{\eta=1}}^{l-\frac{\zeta}{2}} \imath Z_{\bar{j}_\eta} \right)^\dag  \ket{i} \\
    & \times \left( \bigotimes_{\substack{k=1 \\ k~\text{odd}}}^{\zeta-1} p_{i_k,i_{k+1}} O_{i_k} \vec{Z}_{(i_k+1)}^{(i_{k+1}-1)} O_{i_{k+1}} \bigotimes_{\substack{\eta=1}}^{l-\frac{\zeta}{2}} \imath Z_{\bar{j}_\eta}\right) \\
    & = 2^{-n} \sum_{[\bar{j}] \in \binom{n-\zeta}{l-\frac{\zeta}{2}}}\sum_{\substack{O_{i_1} \in \{X_{i_1} Y_{i_1}\} \\ O_{i_2} \in \{X_{i_2},Y_{i_2}\} \\ ... \\ O_{i_\zeta} \in \{X_{i_\zeta},Y_{i_\zeta}\}}} \bigotimes_{\substack{k=1 \\ k~\text{odd}}}^{\zeta-1} \bigotimes_{\substack{\eta=1}}^{l-\frac{\zeta}{2}} \bra{0_{i_k} ... 0_{i_{k+1}}... 0_{\bar{j}_\eta}} O_{i_k} \vec{Z}_{(i_k+1)}^{(i_{k+1}-1)} O_{i_{k+1}}  Z_{\bar{j}_\eta} \ket{1_{i_k}0 ...0 1_{i_{k+1}}0...0_{\bar{j}_\eta}} \\
    & \times \left( O_{i_k} \vec{Z}_{(i_k+1)}^{(i_{k+1}-1)} O_{i_{k+1}} Z_{\bar{j}_\eta}  \right) \\
    & = 2^{-n} \left( \bigotimes_{\substack{k=1 \\ k~\text{odd}}}^{\zeta-1} \sum_{\substack{O_{i_k} \in \{X_{i_k} Y_{i_k}\} \\ O_{i_{k+1}} \in \{X_{i_{k+1}},Y_{i_{k+1}}\} }} \bra{0_{i_k} 0_{i_{k+1}}} O_{i_k} O_{i_{k+1}}   \ket{1_{i_k}1_{i_{k+1}}} O_{i_k} \vec{Z}_{(i_k+1)}^{(i_{k+1}-1)} O_{i_{k+1}} \right)  \left( \sum_{[\bar{j}] \in \binom{n-\zeta}{l-\frac{\zeta}{2}}} \bigotimes_{\substack{\eta=1}}^{l-\frac{\zeta}{2}} Z_{\bar{j}_\eta} \right) 
\end{split}
\end{equation}
where in the first term of the final line we have swapped the order of the sum and tensor product, such that we are essentially considering $\zeta/2$ copies of states with Hamming weight two. Using the results of Lemma~\ref{Lemma:TwoElectronEval}, the first term in the final line above evaluates to 
\begin{equation}
    \bigotimes_{\substack{k=1 \\ k~\text{odd}}}^{\zeta-1} 2^2 \ket{1}\bra{0}_{i_k} \otimes \ket{1}\bra{0}_{i_{k+1}} \otimes \vec{Z}_{(i_k+1)}^{(i_{k+1}-1)}  = 2^\zeta \ket{i}_{[i]}\bra{\mathbf{0}}_{[i]} \bigotimes_{\substack{k=1 \\ k~\text{odd}}}^{\zeta-1}\vec{Z}_{(i_k+1)}^{(i_{k+1}-1)}
\end{equation}
where we use $\ket{x}_{[i]}$ to denote the qubits of computational basis state $\ket{x}$ corresponding to the set $[i]$. Hence, denoting $|y|$ as the Hamming weight of a computational basis state $\ket{y}$,
\begin{equation}
\begin{split}
    \mathcal{P}_{0, \zeta} \left[ \Pi_{2l} \left( \ket{i}\bra{\mathbf{0}} \right) \right] & = 2^{\zeta-n} \sum_{[\bar{j}] \in \binom{n-\zeta}{l-\frac{\zeta}{2}}} \mathcal{P}_{0, \zeta} \left[\ket{i}_{[i]}\bra{\mathbf{0}}_{[i]} \bigotimes_{\substack{k=1 \\ k~\text{odd}}}^{\zeta-1}\vec{Z}_{(i_k+1)}^{(i_{k+1}-1)} \bigotimes_{\substack{\eta=1}}^{l-\frac{\zeta}{2}} Z_{\bar{j}_\eta} \right] \\
    & = 2^{\zeta-n} \sum_{[\bar{j}] \in \binom{n-\zeta}{l-\frac{\zeta}{2}}} \sum_{|x|,|y|=\{0,\zeta\}} \ket{x}\bra{x} \left( \ket{i}_{[i]}\bra{\mathbf{0}}_{[i]} \bigotimes_{\substack{k=1 \\ k~\text{odd}}}^{\zeta-1}\vec{Z}_{(i_k+1)}^{(i_{k+1}-1)} \bigotimes_{\substack{\eta=1}}^{l-\frac{\zeta}{2}} Z_{\bar{j}_\eta} \right) \ket{y}\bra{y} \\
    & = 2^{\zeta-n} \sum_{[\bar{j}] \in \binom{n-\zeta}{l-\frac{\zeta}{2}}} \ket{i}\bra{\mathbf{0}}\\
    & = 2^{\zeta-n} \binom{n-\zeta}{l - \frac{\zeta}{2}} \ket{i}\bra{\mathbf{0}}
\end{split}
\end{equation}
\end{widetext}
Thus for $\frac{\zeta}{2} \leq l \leq n - \frac{\zeta}{2}$, $\mathcal{P}_{0, \zeta} \left[ \Pi_{2l} \left( \ket{i}\bra{\mathbf{0}} \right) \right] = 2^{\zeta-n} \binom{n-\zeta}{l - \frac{\zeta}{2}} \ket{i}\bra{\mathbf{0}}$,  as required.
\end{proof}

\begin{lemma}\label{Lemma:TwoElectronEval}
    The sum
    \begin{equation}
        \sum_{\substack{O_{i} \in [X_{i} Y_{i}] \\ O_{j} \in [X_{j},Y_{j}]}} \bra{0_{i} 0_{j}} O_{i} O_{j} \ket{1_{i} 1_{j}} O_{i} \vec{Z}_{i+1}^{j-1} O_{j}
    \end{equation}
    evaluates to
    \begin{equation}
        2^2 \ket{1}\bra{0}_{i} \otimes \ket{1}\bra{0}_{j} \otimes \vec{Z}_{i+1}^{j-1}
    \end{equation}
    where $\vec{Z}_a^b := \bigotimes_{k=a}^b Z_k$.
\end{lemma}
\begin{proof}
    Explicitly computing the terms in the sum yields
    \begin{equation}
        \left(X_i X_j - Y_i Y_j - \imath(X_iY_j + Y_iX_j)\right)\vec{Z}_{i+1}^{j-1}. 
    \end{equation}
    Using the decomposition $X = \left(\ket{0}\bra{1} + \ket{1}\bra{0} \right)$ and $Y = \left(-\imath \ket{0}\bra{1} + \imath \ket{1}\bra{0}\right)$ yields
    \begin{equation}
        2^2 \ket{1}\bra{0}_{i} \otimes \ket{1}\bra{0}_{j} \otimes \vec{Z}_{i+1}^{j-1}
    \end{equation}
    as required.
\end{proof}

\section{Robust Matchgate shadows}\label{app:robust_shadow} 

In this section, we formulate the robust Matchgate shadow protocol, providing an alternative derivation from those in Refs.~\cite{zhao2023group,wu2023error}. 

\subsection{Robust shadow protocol}
In this part, we review the robust shadow estimation~\cite{chen2021robust}, which serves as a basis for robust Matchgate shadows. Eq.~\ref{eq:quantum_channel} is particularly useful when it comes to noise, where the channel is modified by noise as
\begin{equation}
    \widetilde{\mathcal{M}} = \sum_{\lambda \in R_G} \widetilde{f}_{\lambda}\Pi_{\lambda},\;\;\;\; \widetilde{f}_{\lambda} = \frac{\text{tr}(\mathcal{M}_Z \Lambda \Pi_{\lambda})}{\text{tr}(\Pi_{\lambda})}, \label{eq:noisy_channel}
\end{equation}
where $\widetilde{f}_{\lambda}$ is the noisy eigenvalue, $\Lambda$ denotes the noise channel, and we've assumed the noise is gate-independent, time-stationary, and Markovian (GTM)~\cite{chen2021robust}. Within this assumption, the noise-free expectation values of $\text{tr}(O_i\rho)$ could be retrieved by inverting this noisy channel, if $\rho$ is perfectly prepared on quantum computers. To achieve that, $\{\widetilde{f}_{\lambda}\}$ need to be determined first, and different measurement schemes have been designed for that according to the specific group $G$ used. Such a protocol is known as robust shadow estimation~\cite{chen2021robust}.

\subsection{Noisy Matchgate channel}
In this section we provide derivations for the recently introduced robust Matchgate shadows approach~\cite{zhao2023group, wu2023error}. Our derivations closely follow previous work on Matchgate randomized benchmarking~\cite{helsen2022matchgate}, highlighting the close connection between these two topics.

The robust shadow protocol assumes the quantum noise satisfies the GTM assumption, in which case the measurement channel is only modified in the eigenvalues associated with each projector. The idea is therefore to determine those noisy coefficients $\widetilde{f}_{\lambda}$ and apply the noisy inverse channel $\widetilde{\mathcal{M}}^{-1}$ when computing the expectation values. For the noisy case, we can generalize the derivation for eigenvalues in the noiseless case into
\begin{widetext}
\begin{equation}
\begin{split}
    \widetilde{f}_{2l} = \frac{\text{Tr}\left[\mathcal{M}_Z \Lambda \Pi_{2l} \right]}{\text{Tr}\left[\Pi_{2l}\right]} & = \begin{pmatrix}
        2n\\2l
    \end{pmatrix}^{-1} \sum_{S\in\left(\begin{smallmatrix}
        [2n]\\ 2l
    \end{smallmatrix}\right)}\sum_{b\in\{0,1\}^n} \llangle b|\Lambda|\gamma_S\rrangle\llangle \gamma_S|b\rrangle\\
    & = \begin{pmatrix}
        2n\\2l
    \end{pmatrix}^{-1} \sum_{b\in\{0,1\}^n} \sum_{T\in \left(\begin{smallmatrix}
        [n]\\l
    \end{smallmatrix}\right)} \left|\frac{(-i)^{|T|}}{\sqrt{2^n}} (-1)^{\sum_{j\in T} b_j}\right|^2 \llangle \gamma_{\text{pairs}(T)}|\Lambda|\gamma_{\text{pairs}(T)}\rrangle\\
    & = \begin{pmatrix}
        2n\\ 2l
    \end{pmatrix}^{-1} \sum_{T\in \left(\begin{smallmatrix}
        [n]\\l
    \end{smallmatrix}\right)} \llangle \gamma_{\text{pairs}(T)}|\Lambda|\gamma_{\text{pairs}(T)}\rrangle. \label{eq:noisy_matchgate_coeff}
\end{split}    
\end{equation}
\end{widetext}
To get $\widetilde{f}_{2l}$, we need to determine $\llangle \gamma_{\text{pairs}(T)}|\Lambda |\gamma_{\text{pairs}(T)}\rrangle$, which corresponds to the diagonal elements of the Liouville representation of noise channel $\Lambda$ in the $|\gamma_S\rrangle$ basis.

Interestingly, a recent work on randomized benchmarking~\cite{helsen2022matchgate} proposed an efficient way to characterize the noise in the quantum hardware by defining and computing a set of Majorana fidelities $\lambda_k$ using random Matchgate circuits. Here we reformulate some of their derivations using our established conventions and get more insights into the interpretation and connection between $\lambda_k$ and $\widetilde{f}_{2l}$.

The protocol consists of multiple rounds with varying parameters $k \in [2n]$ and circuit sequence lengths $m$. Due to our focus on fermionic quantum simulations, we'll only discuss the case when $k$ is even. Each round starts with the preparation of the all-zero state $|\mathbf{0}\rangle$. This is followed by $m$ Matchgate unitaries $U_{Q_1}, \dots, U_{Q_m}$, chosen uniformly and independently at random. Finally, all qubits are measured in the $Z$-basis. By averaging over many random sequences, we obtain an estimate $\hat{c}_{2l}(m)$ of the weighted average
\begin{equation}
    c_{2l}(m) = \frac{1}{|B(2n)|^m} \sum_{\{Q\} \in B(2n)} \sum_{b\in\{0,1\}^n} \alpha_{2l}(b, Q) p(b|Q, m), \label{eq:c_2l}
\end{equation}
where $Q = Q_m\dots Q_1$ and $p(b|Q, m)$ represents the probability of measurement outcome $|b\rangle$. A correlation function $\alpha_{2l}$ has also been defined as
\begin{equation}
    \alpha_{2l}(b, Q) = \text{tr}\left[|b\rangle\langle b|\Pi_{2l}\left(U_Q \rho_0 U^{\dagger}_Q\right)\right]. \label{eq:alpha_2l}
\end{equation}
We estimate $p(b|Q,m)$ from the measurement outcomes received from the quantum computer, whereas we compute $\alpha_{2l}(b,Q)$ analytically. Assuming the state preparation and measurement (SPAM) are perfect, we have
\begin{widetext}
\begin{equation}
\begin{split}
    c_{2l}(m) & = \frac{1}{|B(2n)|^m} \sum_{\{Q\} \in B(2n)}\sum_{b\in\{0,1\}^n} \alpha_{2l}(b, Q_m \dots Q_1) \llangle b|\Lambda \mathcal{U}_{Q_m} \dots \Lambda \mathcal{U}_{Q_1} |\mathbf{0}\rrangle\\
    & = \frac{1}{|B(2n)|^m} \sum_{\{Q\} \in B(2n)} \sum_{b\in\{0,1\}^n} \llangle b \otimes b|(\Pi_{2l}\otimes \Lambda) \mathcal{U}^{\otimes 2}_{Q_m} (\mathcal{I} \otimes \Lambda) \mathcal{U}^{\otimes 2}_{Q_{m-1}} \dots (\mathcal{I} \otimes \Lambda) \mathcal{U}^{\otimes 2}_{Q_1} |\mathbf{0}\otimes \mathbf{0}\rrangle\\
    & = \sum_{b\in\{0,1\}^n} \llangle b^{\otimes 2}| (\Pi_{2l}\otimes \Lambda) \left[\sum_{k'=0}^{2n} |\Upsilon_{k'}^{(2)}\rrangle\llangle \Upsilon_{k'}^{(2)}| (\mathcal{I} \otimes \Lambda)\right]^{m-1} \sum_{k'=0}^{2n} |\Upsilon_{k'}^{(2)}\rrangle\llangle \Upsilon_{k'}^{(2)} |\mathbf{0}^{\otimes 2}\rrangle\\
    & = \sum_{b\in\{0,1\}^n} \llangle b^{\otimes 2} |(\mathcal{I} \otimes \Lambda)|\Upsilon_k^{(2)}\rrangle \left[\llangle \Upsilon_k^{(2)}| (\mathcal{I} \otimes \Lambda)|\Upsilon_k^{(2)}\rrangle\right]^{m-1} \llangle \Upsilon_k^{(2)} |\mathbf{0}^{\otimes 2}\rrangle. \label{eq:c_2l(m)}
\end{split}
\end{equation}
\end{widetext}
Ref.~\cite{helsen2022matchgate} defined the term $\llangle \Upsilon_k^{(2)}| (\mathcal{I} \otimes \Lambda)| \Upsilon_k^{(2)}\rrangle$ as the Majorana fidelities $\lambda_{2l}$. We also note that the first term in the last line is exactly what we've derived in Eq.~\ref{eq:noisy_matchgate_coeff} if we assume SPAM error-free. Therefore, we choose $m=1$ and have
\begin{equation}
\begin{split}
    & c_{2l}(1) = \sum_{b\in\{0,1\}^n} \llangle b^{\otimes 2} |(\mathcal{I} \otimes \Lambda)|\Upsilon_k^{(2)}\rrangle \llangle \Upsilon_k^{(2)} |\mathbf{0}^{\otimes 2}\rrangle\\
    & = \begin{pmatrix}
        2n\\ 2l
    \end{pmatrix}^{-1} \sum_{b\in\{0,1\}^n} \sum_{S,S' \in \left(\begin{smallmatrix}
        [2n]\\ 2l
    \end{smallmatrix}\right)} \llangle b|\gamma_S\rrangle \llangle b|\Lambda|\gamma_S\rrangle \llangle \gamma_{S'}|\mathbf{0}\rrangle^2\\
    & = \frac{1}{2^n} \begin{pmatrix}
        2n \\ 2l
    \end{pmatrix}^{-1} 
    \begin{pmatrix}
        n \\ l
    \end{pmatrix} \sum_{T\in \left(\begin{smallmatrix}
        [n]\\ l
    \end{smallmatrix}\right)} \llangle \gamma_{\text{pairs}(T)}|\Lambda|\gamma_{\text{pairs}(T)}\rrangle\\
    & = \frac{1}{2^n} \begin{pmatrix}
        n \\ l
    \end{pmatrix} \widetilde{f}_{2l}.
\end{split}
\end{equation}
We now have established the relationship between $c_{2l}$ and $\widetilde{f}_{2l}$. Therefore the noisy eigenvalues can be extracted either by fitting $c_{2l}(m)$ or simply measuring $c_{2l}(1)$. For $m$ other than 1, we have:
\begin{equation}
    c_{2l}(m) = c_{2l}(1) \left[\begin{pmatrix}
        2n\\ 2l
    \end{pmatrix}^{-1} \sum_{S \in \left(\begin{smallmatrix}
        [2n]\\ 2l
    \end{smallmatrix}\right)} \llangle \gamma_S|\Lambda|\gamma_S\rrangle\right]^{(m-1)},
\end{equation}
which provides a microscopic origin for $\lambda_{2l}$.

In actual hardware experiments, SPAM error is unavoidable. For measurement error, it satisfies the GTM assumption for noisy shadows and therefore it's not a concern. For state preparation error, its effect is absorbed into the $c_{2l}(1)$ prefactor~\cite{helsen2022matchgate} above, by replacing $|\mathbf{0}\rrangle$ with $|\widetilde{\mathbf{0}}\rrangle$ in the first equality of Eq.~\ref{eq:c_2l(m)}. We further note that the state preparation error in the shadow experiments, i.e., the noise effects in preparing $\rho$ can also be taken into account in the above framework, as long as we could mimic as much as possible the state preparation circuit in determining $\widetilde{f}_{2l}$. This sets the motivation for the revised robust Matchgate shadow protocol introduced in Sec.~\ref{subsection:robust_shadow}. \textcolor{black}{We note that if the state preparation noise channel $\Lambda'$ commutes with the twirl $\mathcal{M}$, i.e., $\Lambda'\circ \mathcal{M} = \mathcal{M} \circ \Lambda'$, our revised scheme would correct the state preparation error perfectly, assuming $\Lambda'$ satisfies the GTM assumption. However, this condition may not be strictly satisfied in hardware.}

\subsection{Variance bound}
The variance of $c_{2l}(m)$, defined in Eq.~\ref{eq:c_2l}, is important to determine the efficiency of the robust protocol. A detailed derivation can be found in the Appendix of Ref.~\cite{helsen2022matchgate, wu2023error}, assuming no noise.

We numerically test the sampling complexity to determine $\widetilde{f}_{2l}$, as shown in Fig.~\ref{fig:calibration_sampling_complexity}, and find that it scales closely to $O(n^{2.95})$. The complexity comes from bounding the variance relative to $f_{n}$, the smallest coefficient among all the eigenvalues in the $n$-qubit Matchgate channel.

\begin{figure}[hbt!]
    \centering
    \includegraphics[width=0.45\textwidth]{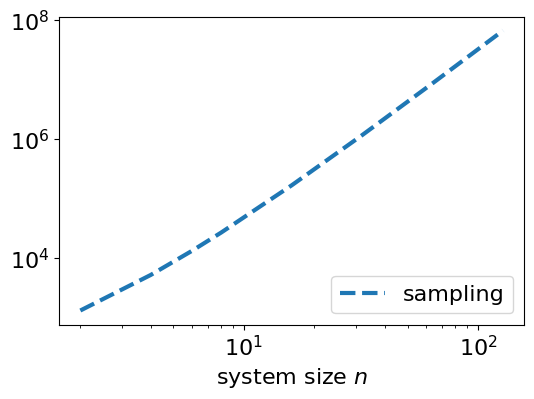}
    \caption{Sampling complexity for determining the Matchgate channel coefficients (with $1\%$ error threshold) scales closely to $O(n^{2.95})$ in the system size limit.}
    \label{fig:calibration_sampling_complexity}
\end{figure}

\section{Implementation of Matchgate shadows}  \label{app:implementation}

The Matchgate circuits are constructed using the open source code located at~\cite{Zhao2021Github}, which requires $O(n^2)$ gates implemented in $O(n)$ depth. Note that the current implementation assumes an even number of particles in the system studied. A generalization into odd-particle cases can be addressed with ancillary qubits, as discussed in App. A of Ref.~\cite{wan2023matchgate}.

\section{Quantum defect embedding theory} \label{app:qdet}

Quantum embedding theories~\cite{sun2016quantum} are frameworks to solve the time-independent Schr\"odinger equation for a system of electrons by separating the problem into the calculation of the energy levels or density of an active space and those of the remaining environment. Each part of the system is described at the quantum-mechanical level, with the active space being treated with a more accurate and computationally more expensive theoretical method than the environment.

Spin defects, e.g., NV center in diamond, are particular suitable for an embedding description as they can be naturally partitioned into the defect center and the host material. Recently, a Green's function-based quantum embedding theory was proposed for the calculation of defect properties, denoted as quantum defect embedding theory (QDET). We'll provide a high-level description of this method in this section, and we direct interested readers to Refs.~\cite{ma2021quantum, sheng2022green} for more details. QDET initiates from a mean-field electronic structure calculation, typically DFT, on a supercell (with hundreds/thousands of atoms) representing the point-defect of interest hosted in a pristine solid. Then an active space is defined by a sub-set of single particle orbitals localized around the defect center, whose excitations are described by an effective Hamiltonian $H_{\text{eff}}$. The effective potential entering $H_{\text{eff}}$ is evaluated by computing the effect of the environment onto the active space with many-body perturbation theory techniques~\cite{sheng2022green}. The Hamiltonian can be solved either using a full configuration interaction (FCI) approach on classical computers, or using quantum algorithms on quantum computers.

In essence, using QDET allows one to reduce the complexity of evaluating many-body states of a small guest region embedded in a large host system: the problem is reduced to diagonalizing a many-body Hamiltonian simply defined on an active space, where the number of degrees of freedom is much smaller than that required to describe the entire supercell of hundreds of atoms.

\section{Quantum simulation details} \label{app:quantum_simulations}

In this section, we provide detailed information on quantum simulations performed in this work, on both quantum hardware platforms and on simulators that can emulate the effects of noise.

\subsection{IBM Q experiments for hydrogen}
We study the hydrogen molecule with five different bond distances, ranging from $0.75 \text{\r{A}}$ to $2.75 \text{\r{A}}$ using the minimal STO-3G basis set. Restricted HF orbitals are used as inputs for subsequent VQE calculation on a noisy quantum simulator (to obtain the quantum trial states) and QC-QMC calculation on quantum hardware. We used 4 qubits on IBM Quantum Hanoi superconducting qubit device, allocated based on the Jordan-Wigner mapping of spin-orbitals to qubits.

\subsubsection{Quantum trial state}
As we introduced in the main text, the quantum trial state adopted in the QC-QMC scheme has to satisfy two requirements: i). $V_T|\Psi_I\rangle = |\Psi_T\rangle$; ii). $V_T|\mathbf{0}\rangle = |\mathbf{0}\rangle$. This implies that the trial circuit should conserve the particle number of the initial states at least in the $\{0, \zeta\}$-particle subspace. We therefore chose the UCCSD ansatz. The issue with conventional UCCSD ansatz is that the circuit is too deep, with a $O(n^4)$ two-qubit gate count, and this is usually beyond the capabilities of NISQ hardware. We, therefore, designed an alternative circuit, shown in Fig.~\ref{fig:hydrogen_shadow_circuit}, that constructs effectively a double excitation. Such a circuit assumes only linear connectivity of the underlying hardware and doesn't require any swap of the physical qubits. It is worth re-emphasizing that this tailored UCCSD ansatz was chosen specifically for the hydrogen molecule, to be sufficient for enabling our QC-QMC experiments, and does not represent a universal solution. We leave the important task for constructing a general and efficient quantum trial state for future investigations.

\begin{figure}[hbt!]
    \centering
    \includegraphics[width=0.45\textwidth]{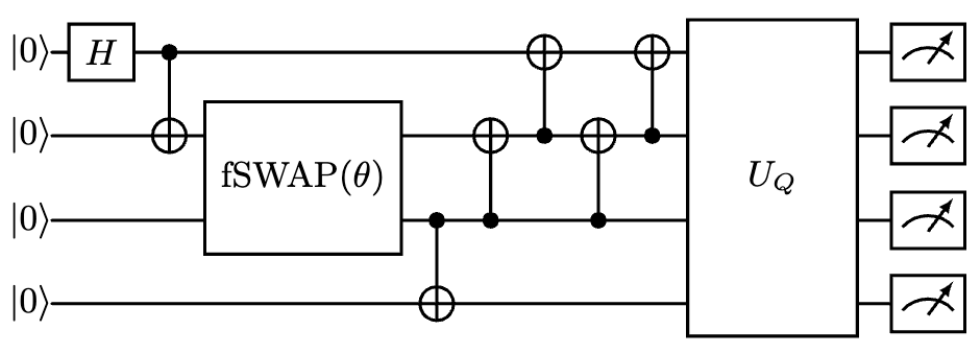}
    \caption{Quantum circuit for the performing Matchgate shadows of hydrogen. The first Hadamard and CNOT gate generates an equal superposition of $|0\rangle$ and $|\Psi_I\rangle$. They are followed by the effective UCCSD ansatz, and the Matchgate circuit $U_Q$ for shadow tomography.}
    \label{fig:hydrogen_shadow_circuit}
\end{figure}

\subsubsection{Hardware information}
The information of the physical qubits used on IBM Q Hanoi device is recorded in Table.~\ref{tab:IBM_Q_hardware}.

\begin{table*}
\caption{\label{tab:IBM_Q_hardware} IBM Q Hanoi hardware information, with data taken from \href{https://quantum.ibm.com/}{IBM Quantum Platform}.}
\begin{ruledtabular}
\begin{tabular}{cccc}
 & Oct.1st & Oct.2nd & Oct.14th \\ 
\hline
physical qubits & [1,2,3,5] & [1,2,3,5] & [6,7,10,12]\\
$T_1 \;[\mu s]$ & [132, 139, 156, 159] & [185, 114, 129, 130] & [137, 101, 112, 227]\\
$T_2 \;[\mu s]$ & [145, 266, 35, 199] & [175, 164, 33, 198] & [157, 252, 215, 262]\\
CNOT error rate $(10^{-3})$ & [3.25, 11.6, 4.71] & [3.20, 8.26, 6.38] & [7.58, 5.58, 6.30]\\
X error rate $(10^{-4})$ & [4.0, 1.3,  2.1, 1.8] & [3.4, 1.8, 2.0, 2.2] & [2.5, 2.2, 2.4, 1.8]\\
Readout error rate $(10^{-2})$ & [2.1, 1.7, 1.0, 0.07] & [1.1, 1.2, 1.1, 0.08] & [2.0, 1.2, 0.09, 1.6]\\
\end{tabular}
\end{ruledtabular}
\end{table*}

\subsubsection{Robust shadow experiments}
The noise determination phase (to measure $\widetilde{f}_{2l}$) and standard shadow experiments are performed employing 16000 circuits each on the same set of physical qubits on the same day, to ensure the noise condition stays the same. Note that the original proposal of shadow tomography~\cite{huang2020predicting} assumes using a single shot for each shadow circuit. For better use of the quantum resources, we followed Ref.~\cite{zhou2023performance} and all the circuits are executed with 1024 shots. The circuit used for determining $\widetilde{f}_{2l}$ only differs in the absence of the first Hadamard gate from the Matchgate shadow circuits used to compute the overlap. This construction ensures that the noise in preparing $\rho$ for QC-AFQMC algorithm is captured as much as possible in computing $\widetilde{f}_{2l}$. The measured values of $\widetilde{f}_{2l}$ are summarized in Table.~\ref{tab:tilde_f}.

\begin{table}
\caption{\label{tab:tilde_f} Data of the noisy eigenvalues of the 4-qubit Matchgate channel, with experiments performed on the IBM Q Hanoi device on Oct 1st, 2nd and 14th. These noisy eigenvalues are used for the robust shadow protocol considering the state preparation error (SP). Specifically, shadows of hydrogen bond distance 0.75 \text{\r{A}} are collected on Oct 14th; shadows for 2.25 \text{\r{A}} are obtained on Oct 1st while the shadows for the rest three are obtained on Oct 2nd.}
\begin{ruledtabular}
\begin{tabular}{cccc}
IBM Q Hanoi & Oct.1st & Oct.2nd & Oct.14th \\ 
\hline
physical qubits & [1,2,3,5] & [1,2,3,5] & [6,7,10,12]\\
$\widetilde{f}_0$ & 1.0 & 1.0 & 1.0\\
$\widetilde{f}_2$ & 0.1083 & 0.1146 & 0.1030\\
$\widetilde{f}_4$ & 0.0629 & 0.0649 & 0.0565\\
$\widetilde{f}_6$ & 0.0961 & 0.1042 & 0.0940\\
$\widetilde{f}_8$ & 0.6461 & 0.7121 & 0.6417\\
\end{tabular}
\end{ruledtabular}
\end{table}

Having determined the values of $\widetilde{f}_{2l}$, one can proceed to apply the robust Matchgate shadows procedure to overlaps, and overlap ratios. In the main text, we plotted the mean absolute error (MAE) of 120 overlap ratios w.r.t. the number of Matchgate shadow circuits used, and observed that the robust protocol was able to mitigate the effects of noise. Here, we plot the MAE of 36 separate overlap ratios, randomly chosen from the total of 120 tested in the main text, shown in Fig.~\ref{fig:MAE_ratio_each}. We can conclude that the noise resilience is present for each trial, and is not an artifact of averaging.

\begin{figure*}[hbt!]
    \centering
    \includegraphics[width=\textwidth]{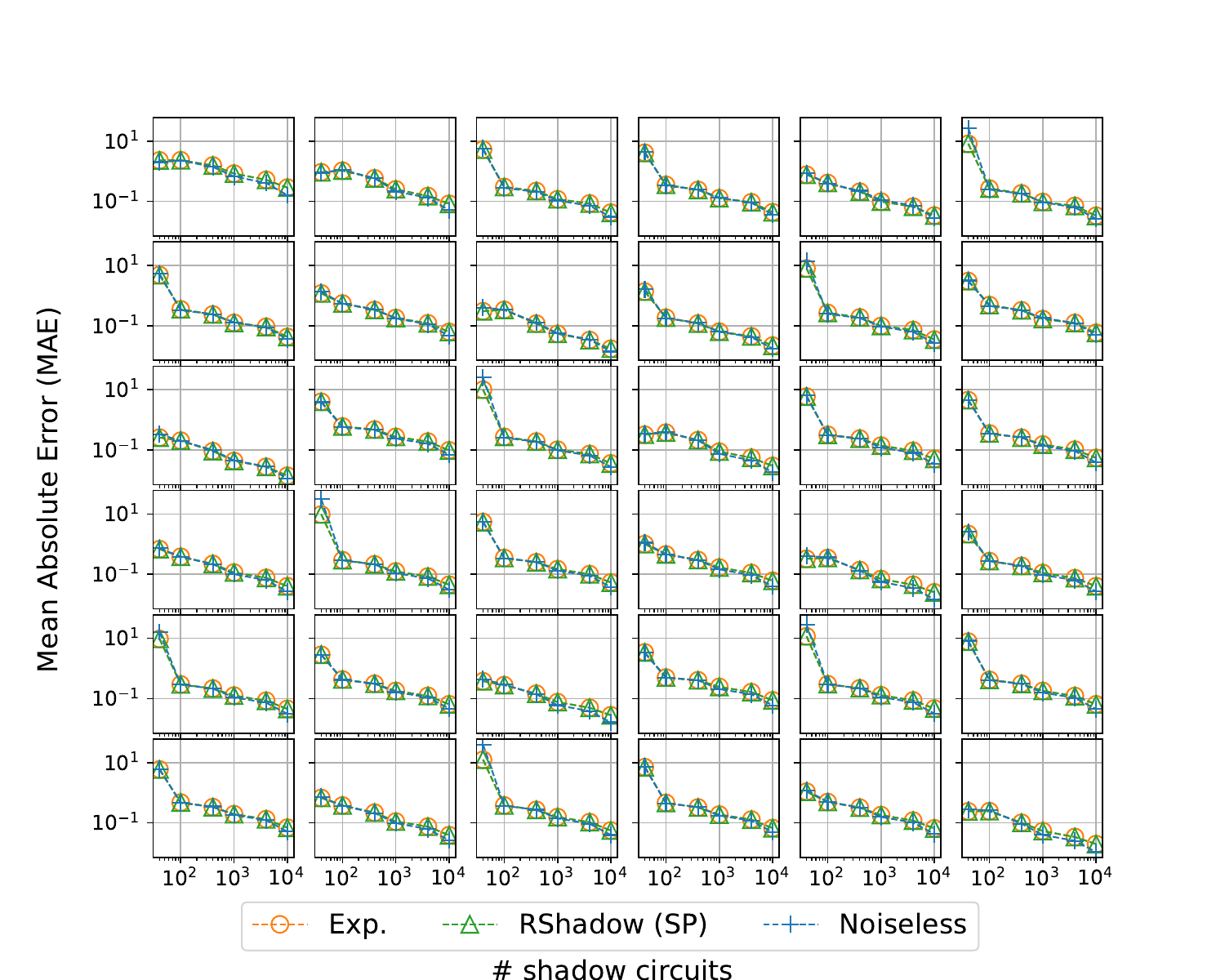}
    \caption{The mean absolute error (MAE) of 36 overlap ratios randomly selected from the tested 120 overlap ratios' pool. For each overlap ratio, noisy results from IBM Hanoi, w/wo the robust shadow correction (orange circle/green triangle) are all in good agreement with the noiseless reference value (blue cross). This verifies the noise resilience discussed in the main text.}
    \label{fig:MAE_ratio_each}
\end{figure*}

\subsubsection{Classical post-processing} \label{app:classical_pp}

As discussed in the main text, in the interest of conserving computational resources, only two of the five data points (solid points) in Fig.~\ref{fig:H2_AFQMC} were obtained using the scalable Matchgate shadows approach. The remaining three data points were obtained using an exponentially scaling approach used in Ref.~\cite{huggins2022unbiasing} that is ultimately more efficient than the scalable Matchgate approach for small system sizes. We verified that the two schemes give the same results for overlap amplitudes. This exponentially scaling approach works by first computing and tabulating the local quantities, i.e., overlap amplitudes, force bias and local energies, for all of the canonical computational basis states (of which there are 4 for hydrogen). The local quantities of arbitrary walker states can then be computed by decomposing the Slater determinant into a linear combination of the 4 basis states of hydrogen, and then taking the corresponding linear combination of the tabulated data. For small molecules like hydrogen, this yields a very efficient simulation. However, for larger molecules an arbitrary Slater determinant would decompose into a number of canonical computational basis states scaling superpolynomially with $n$, rendering this method unscalable.

\subsection{IonQ experiments for NV center}

The effective Hamiltonian describing the strongly correlated electronic states of the NV center is generated by QDET. The detailed information of classical calculations is documented in Ref.~\cite{huang2023quantum}. The NV center in diamond is simulated using 4 qubits on the IonQ Aria trapped ion quantum computer.

\subsubsection{Quantum trial state}

Since the ground state of NV center is an equal superposition of two Slater determinants~\cite{huang2022simulating}, we can generalize the quantum trial state we used for the hydrogen molecule and apply it to the NV case. The circuit is shown in Fig.~\ref{fig:nv_shadow_circuit}.

\begin{figure}[hbt!]
    \centering
    \includegraphics[width=0.45\textwidth]{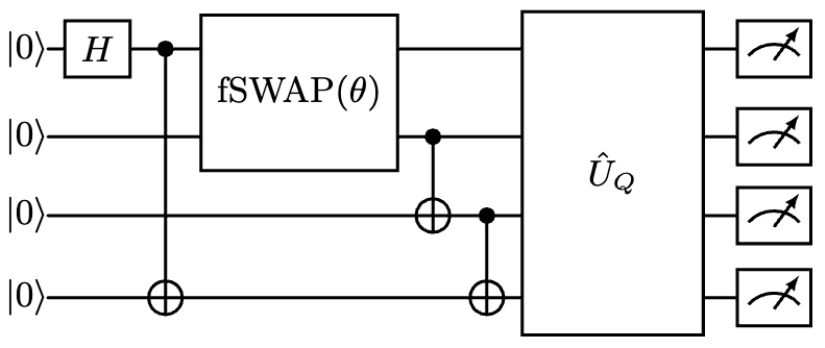}
    \caption{Quantum circuit for the performing Matchgate shadows of NV center. The first Hadamard and CNOT gate generates an equal superposition of $|0\rangle$ and $|\Psi_I\rangle$. They are followed by the effective UCCSD ansatz, and the Matchgate circuit $U_Q$ for shadow tomography.}
    \label{fig:nv_shadow_circuit}
\end{figure}

\subsubsection{Hardware information}

The IonQ Aria device has an average of 99.97\% single-qubit gate fidelity and 98.6\% two-qubit gate fidelity, with all-to-all qubit connectivity. The average readout fidelity is 99.55\%. Detailed information about the device can be found on the \href{https://ionq.com/}{IonQ website}.

\subsection{Noisy simulations of the water molecule}
The water molecule is a more strongly correlated problem than hydrogen, with more than one correlated pair of electrons. We simulated it at equilibrium geometry with a (4e, 4o) active space based on the restricted Hartree Fock canonical orbitals, following Ref.~\cite{ryabinkin2018qubit}, using 8 qubits on the Pennylane quantum simulator~\cite{bergholm2018pennylane} with various noise models. The 8 qubits are associated with each orbital according to ascending order in energy with spin-up and down alternations.

\subsubsection{Quantum trial state}
Since the water molecule has a perfect pairing of electrons in the ground state, we consider a parameterized quantum circuit employing a 4-qubit double-excitation gate $U_{\text{DE}}(\theta)$:
\begin{equation}
    U_{\text{DE}}(\theta) |1100\rangle = \cos\left(\frac{\theta}{2}\right) |1100\rangle - \sin\left(\frac{\theta}{2}\right) |0011\rangle.
\end{equation}
This gate is implemented between qubit indexed $[0,1,4,5],\; [2,3,6,7],\; [1,2,7,4],\; [0,3,6,5]$. These indices are chosen because they construct those Slater determinants with the largest coefficients from a CASCI calculation. This ansatz is then fed to a (noiseless) VQE calculation to produce a wavefunction with an energy error of approximately 10 mHa at the equilibrium geometry. This quantum trial state is only used for demonstrating the noise resilience in evaluating the overlap ratios using noise models, as the full QC-QMC calculations for the water molecule would be too costly to carry out.

\bibliography{bibliography.bib}

\end{document}